\documentclass[final]{eptcs}

\usepackage{etoolbox}
\newtoggle{proofs}
\togglefalse{proofs}    

\iftoggle{proofs}{%
}
{%
\usepackage{amsthm}
}

\usepackage{enumerate}
\usepackage{mathtools} 
\usepackage{amsmath}
\usepackage{amssymb}
\usepackage{eucal}

\usepackage{tikz}
\usetikzlibrary{decorations.pathmorphing} 

\usepackage{cases} 
\usepackage{environ, thmtools, thm-restate}  
\usepackage{bussproofs}
\EnableBpAbbreviations

\usepackage{stmaryrd}
\usepackage{graphicx}
\usepackage{wasysym}
\usepackage{makeidx}                

\usepackage[center]{subfigure}

\usepackage[nomargin,inline,index]{fixme} 
\fxusetheme{color}
\usepackage{ifdraft}                

\usepackage{verbatim} 

\newcommand{\hidden}[1]{}

\newcommand{\irule}[2]{\dfrac{#1}{#2}}

\NewEnviron{reslemma}[2][]
  {\begin{restatable}[#1]{mylem}{#2}%
    \label{lem:#2}
    \BODY
   \end{restatable}%
  }

\iftoggle{proofs}{%
\declaretheorem[name=Notation,numberlike=definition]{notation}
\declaretheorem[name=Theorem,numberlike=theorem]{mythm}
\declaretheorem[name=Lemma,numberlike=lemma]{mylem}
\declaretheorem[name=Theorem,numberwithin=section]{appthm}
\declaretheorem[name=Lemma,numberwithin=section]{applem}
\declaretheorem[name=Definition,numberwithin=section]{appdef}
}
{%
\declaretheorem[name=Theorem,numberwithin=section]{theorem}

\declaretheorem[name=Definition,numberlike=theorem]{definition}
\declaretheorem[name=Example,numberlike=theorem]{example}

\declaretheorem[name=Theorem,numberwithin=section]{appthm}
\declaretheorem[name=Lemma,numberwithin=section]{applem}
\declaretheorem[name=Definition,numberwithin=section]{appdef}
}

\newcommand{\powset}[1]{\wp(#1)}

\newcommand{\subseteqfin}{\subseteq_{\it fin}}

\newcommand{\seq}[1]{\langle {#1} \rangle}
\renewcommand{\epsilon}{\varepsilon}
\newcommand{\mkset}[1]{\overline{#1}}

\newcommand{\mkcredit}[2][]{\Gamma_{#1}({#2})}

\newcommand{\eqdef}{\triangleq}

\newcommand{\setenum}[1]{\{#1\}}
\newcommand{\setcomp}[2]{\{{#1} \mid {#2}\}}

\newcommand{\ES}{\textrm{\textup{ES}}}
\newcommand{\CES}{\textrm{\textup{CES}}}
\newcommand{\esname}{\mathcal{E}}

\newcommand{\CF}[1]{{\it CF}{({#1})}}
\newcommand{\Con}{{\it Con}}

\newcommand{\event}[2]{{\pmv #1}:{\atom #2}}

\newcommand{\cname}{\mathcal{C}}
\newcommand{\aname}{\mathcal{A}}
\newcommand{\princsym}{\pi}
\newcommand{\princ}[2][]{\princsym_{#1}({#2})}
\newcommand{\invprinc}[2][]{\princsym_{#1}^{-1}({#2})}

\newcommand{\winsym}{\mathcal{W}}
\newcommand{\win}[1][]{\winsym_{#1}}

\newcommand{\payoffsym}[1][]{\Phi_{#1}}
\newcommand{\payoff}[1][]{\payoffsym[#1]}

\newcommand{\reach}[1][]{\mathcal{R}_{{#1}}}

\newcommand{\preach}[2][]{\reach[{#1}]^{{#2}}}

\newcommand{\enc}[2][]{[{#2}]_{#1}}
\newcommand{\amod}[2][]{{#1}{#2}}
\newcommand{\unmod}[2][]{{#2}^{#1}}

\newcommand{\urgent}[1][]{\mathcal{U}_{{#1}}}

\newcommand{\curgent}[2][]{\urgent[{#1}]^{#2}}

\newcommand{\prudent}[1][]{\mathcal{P}_{{#1}}}

\newcommand{\cprudent}[2][]{\prudent[{#1}]^{#2}}

\newcommand{\primes}[1]{\mathit{atoms}({#1})}

\newcommand{\Real}[1]{\mathrm{Real}}

\newcommand{\pmv}[1]{\ensuremath{\mathsf{#1}}}

\newcommand{\atom}[1]{\textit{#1}}
\newcommand{\atoms}{\mathit{E}}

\newcommand{\nrule}[1]{{\scriptsize\textsc{\scriptsize #1}}}
\newcommand{\sem}[1]{\mbox{\ensuremath{\llbracket{#1}\rrbracket}}}

\newcommand{\imp}{\rightarrow}
\newcommand{\coimp}{\twoheadrightarrow}

\newcommand{\pcl}{\textup{PCL}}

\newcommand*{\vcenteredhbox}[1]{\begingroup%
\setbox0=\hbox{#1}\parbox{\wd0}{\box0}\endgroup}

\newcommand{\mytitle}{Contract agreements via logic}

\title{\mytitle}

\iftoggle{proofs}{%
\author{
Massimo Bartoletti \and Tiziana Cimoli \and Paolo Di Giamberardino 
\\[4pt] 
{\small Universit\`a degli Studi di Cagliari, Italy} 
\\[4pt]
{\small\texttt{\{bart,t.cimoli,digiambe\}@unica.it}}
\\[16pt]
Roberto Zunino  \\[4pt] 
{\small Universit\`a di Trento and COSBI, Italy} 
\\[4pt]
{\small\texttt{roberto.zunino@unitn.it}}
}
\institute{}
}
{%
\author{Massimo Bartoletti \qquad Tiziana Cimoli \qquad Paolo Di Giamberardino
\institute{Dipartimento di Matematica e Informatica, Universit\`a degli Studi di Cagliari, Italy 
}
\and
Roberto Zunino
\institute{Dipartimento di Matematica, Universit\`a degli Studi di Trento and COSBI, Italy
}
}
\def\titlerunning{\mytitle}
\def\authorrunning{M. Bartoletti, T. Cimoli, P. Di Giamberardino, R. Zunino}
}

\begin{document}

\maketitle

\begin{abstract}
We relate two contract models: 
one based on event structures and game theory, 
and the other one based on logic.
In particular, we show that the notions of agreement and  winning strategies
in the game-theoretic model are related to that of provability 
in the logical model.
\end{abstract}

\section{Introduction} \label{sect:introduction}

Contracts are gaining an increasing relevance in the design 
and implementation of concurrent and distributed systems.
This is witnessed by the proliferation of 
proposals of models and standards for contracts 
appeared in the literature in the last few years.
For instance, \emph{choreography languages} like
WS-CDL~\cite{wscdl}, BPEL4Chor~\cite{bpel4chor} and Scribble~\cite{scribble}
can be used to specify the overall interaction protocol
of a set of Web services.
By projecting a choreography on each of the participants,
we obtain the specification of the behaviour expected from
each single service involved in the application.
These projections can be interpreted as contracts: 
if the actual implementation of each Web service respects its contract,
then the overall application is guaranteed to behave correctly;
otherwise, the service violating its contract may be responsible
(and punishable) for the global failure.
On a more theoretical side, formal models for contracts have been devised
by adapting and extending models of concurrent systems, 
such as Petri nets~\cite{Aalst10cj}, 
event structures~\cite{Hildebrandt10places,BCPZ12places}, 
process algebras~\cite{bhty10,Bravetti08tgc,Bravetti07fsen,Castagna09toplas,Honda08popl},
timed automata~\cite{Lomuscio11fi,Raimondi08fse},
and by extending various logics, 
such as modal~\cite{Abadi93access},
intuitionistic~\cite{Abadi93logical,BZ10lics}, linear~\cite{Abadi93logical},
and deontic~\cite{Prisacariu11jlap,Gelati04normative} logics
(just to cite a few recent approaches).

A main motivation for using contracts
resides in the fact that large distributed applications 
are often constructed by dynamically discovering and composing 
services published by different organizations.
The larger an application is, the greater the probability
that some of its components deviates from the expected behaviour
(either because of unintentional bugs, or 
maliciousness of a competing organization).
Hence, it becomes crucial to protect oneself from other participants'
misbehaviour.
Standard enforcement mechanisms do not apply,
because of the total lack of control on code run by 
mutually untrusted, distributed participants.
Instead, contracts may offer protection by legally binding
the participants in a service composition to either behave as
prescribed, or otherwise be blamed for a contract breach~\cite{Armbrust10cacm}.

In this methodology,
contracts are the pillars which support the reliability of
distributed applications, 
hence the choice of the actual contract model to be used is critical.
However, the ecosystem of contract models proposed in the literature 
is wide and heterogeneous, and the actual properties 
and the relations among different models are not clearly established.
In particular, there is a gap between the two main paradigms
for modelling contracts, i.e.\ the one which interprets them 
as interactive multi-agent systems, and the one where contracts
are rendered as formulae of suitable logics.
To contribute towards reducing this gap, 
in this paper we consider two recent models for contracts ---
one based on game-theoretic notions and the other one on logic ---
and we formally relate them.
More precisely, we show that a correspondence exists between 
the fundamental notions in the first model 
(namely, \emph{agreements} and \emph{winning strategies}) and 
provability in the logic-based model.

In the first model~\cite{BCZ13post},
the behaviour of a set of interacting participants is specified 
as a concurrent multi-player game.
The plays of the game are traces of an event structure (\ES)
which models the causal relations among the actions of the participants. 
Intuitively, an enabling $X \vdash e$ in an \ES\
models the fact that the action $e$
becomes an \emph{obligation} after all the actions in $X$ have been performed.
A participant {\pmv A} wins in a play when $(i)$ her \emph{payoff} 
(defined by a given function~$\payoffsym$) is positive in that play,
or $(ii)$ some participant (but not {\pmv A})
can be blamed for a contract violation.
Indeed, if some ${\pmv B} \neq {\pmv A}$ has violated his contract, 
an external judge may eventually provide {\pmv A} with the 
prescribed compensation
(and {\pmv B} with the respective punishment).

Two key notions in this model are that of \emph{agreement} and \emph{protection}.
Intuitively, given a set of contracts,
the agreement property guarantees that each involved participant  
has a \emph{winning} strategy.
Instead, protection is the property of a single contract $\cname$ of {\pmv A}
ensuring that, whenever $\cname$ is composed with any other contract
(possibly that of an adversary),
{\pmv A} has a \emph{non-losing} strategy.
In~\cite{BCZ13post} it is shown that agreement and protection
cannot coexist in a broad class of contracts
where the obligations are modelled as Winskel's \ES~\cite{Winskel86}.
Roughly, to be protected one should wait until the conditions $X$
in some enabling $X \vdash e$
are satisfied before doing the event $e$.
If all participants adhere to this principle, agreement is not possible.
To reconcile agreements with protection, an extension of Winskel's \ES\
has been proposed, which allows for decoupling a conditional promise 
(\text{e.g.}, doing $e$ in change of $X$)
from the temporal order in which events are performed.
In an \emph{\ES\ with circular causality} (\CES\ for short),
an enabling $b \Vdash a$ means that 
``{\pmv A} will do {\atom a} if {\pmv B} \emph{promises} to do {\atom b}''.
This contract protects {\pmv A}, and
when composed with the contract $a \Vdash b$ of {\pmv B},
it admits an agreement.
More in general, in~\cite{BCZ13post} a technique is proposed 
which, given the participants payoffs, synthesises a set of contracts
which guarantee both agreement and protection.

The second model we consider is an extension of
intuitionistic propositional logic (IPC), 
called Propositional Contract Logic (\pcl~\cite{BZ10lics}).
\pcl\ features a ``contractual'' form of implication, 
denoted by~$\coimp$.
The intuition is that a formula $p \coimp q$ entails $q$ 
not only when $p$ is provable, like standard intuitionistic implication, 
but also in the case that a ``compatible'' formula is assumed.
This compatible formula can take different forms, but the archetypal
example is the (somewhat dual) $q \coimp p$.
While $(p \imp q) \land (q \imp p) \imp p \land q$
is \emph{not} a theorem of IPC,
$(p \coimp q) \land (q \coimp p) \imp p \land q$
is a theorem of \pcl.
The logic \pcl\ is decidable~\cite{BZ10lics}.

A first observation about these two models is that
they both allow for a form of ``circular'' assume-guarantee reasoning.
Consider, for example, a participant 
{\pmv A} which promises to do $a$ provided that she receives $b$ in exchange, 
and a participant {\pmv B} which, dually, 
promises to do $b$ in exchange of $a$.
In the game-theoretic model, these obligations are represented by 
a \CES\ with enablings $b \Vdash a$ and $a \Vdash b$.
Given the intended payoff functions, this contract admits an agreement.
The winning strategies of {\pmv A} and {\pmv B} 
prescribe both participants to do their events 
(without waiting for the other to take the first step), 
so leading
to a \emph{configuration} $\setenum{a,b}$ of the \CES.
In the logical model, the scenario above is represented by the
\pcl\ formula $(b \coimp a) \land (a \coimp b)$.
As noted above, this formula entails both $a$ and~$b$ in the proof system
of \pcl.
Hence, a connection seems to exist between the agreement property
in the game-theoretic model and provability in \pcl.

A main contribution of this paper is to formalise this connection.
More precisely, Theorem~\ref{th:ces-pcl:agreement} shows
that agreement in conflict-free contracts
corresponds to provability in Horn \pcl\ theories.
This correspondence has an important consequence,
since it provides us with a polynomial algorithm 
for provability in Horn \pcl\
(in contrast with the fact that provability in \emph{full} \pcl\ 
is PSPACE-complete, as well as in IPC and 
in its implicational fragment~\cite{Statman79pspace}).
We illustrate this point with the help of some examples
(Ex.~\ref{ex:double-diamond} and~\ref{ex:shy-dancers})
where we show that apparently hard questions in \pcl\ admit an easy
answer when passing to the realm of contracts.

We deepen the above-mentioned correspondence by relating
winning strategies for the game-theoretic contracts with
proofs in \pcl.
The idea is that a proof in the logic
induces an ordering among the atoms.
For instance, to use the elimination rule of $\imp$
in a proof of $\Delta, a \imp b \vdash b$, 
one must first construct a proof of $a$
(similarly to the ordering imposed by an enabling $a \vdash b$),
whereas in a proof of $\Delta, a \coimp b \vdash b$ the proofs of $a$ and $b$
can be interleaved 
(i.e.\ $a$ can be proved after $b$, similarly to 
the fact that $a \Vdash b$ allows $a$ to be done after $b$).
We introduce in Section~\ref{proof-traces} the notion of \emph{proof traces},
that represent
the sequences of atoms respecting the order imposed by proofs in \pcl.
Theorem~\ref{th:ces-pcl:prudence-prooftrace} states that 
proof traces correspond, in the contracts realm, 
to the plays where all participants are innocent.
Since these plays can be constructed with a polynomial algorithm,
this result is significant, because it allows for 
performing a non-trivial task in Horn \pcl\
(i.e., constructing proof traces),
through an easier one in contracts.
Finally, Theorem~\ref{th:ces-pcl:winning-strategy} establishes that,
whenever a contract admits an agreement,
proof traces can be projected to winning strategies for all participants.

\iftoggle{proofs}{%
}{
\smallskip
Because of space constraints, the proofs of our results  
are available in~\cite{ces-pcl-long}.
}

\section{Background}
  \subsection{Contracts}

We briefly review the theory of contracts
introduced in~\cite{BCZ13post}. 
A contract is a concurrent system
featuring \emph{obligations} (what I must do in a given state) 
and \emph{objectives} (what I wish to obtain in a given state).

Obligations are modelled as event structures with circular causality (\CES). 
A comprehensive account of \CES\ is in~\cite{BCPZ13fi};
here we shall only recall the needed definitions.
Assume a denumerable universe of atomic actions $a, b, e, \ldots \in E$, 
called \emph{events}, uniquely associated to 
\emph{participants} ${\pmv A}, {\pmv B}, \ldots \in \aname$
by a function $\princsym: E \rightarrow \aname$.
We denote with $\# \subseteq E \times E$ a \emph{conflict} relation
between events, namely if $a \# b$ then $a$ and $b$ cannot occur
in the same computation.
For a set $X \subseteq E$, the predicate $\CF{X}$ is true iff $X$ is 
\emph{conflict-free}, \text{i.e.} $\forall e,e' \in X: \neg (e \# e')$.
We denote with \Con~ the set $\setcomp{X \subseteqfin E}{CF(X)}$.

\begin{definition}[\CES]
A \CES\  $\esname$ is a triple $\seq{\#, \vdash, \Vdash}$, where
\begin{itemize}

\item $\# \;\subseteq E \times E $ is an irreflexive and symmetric 
\emph{conflict} relation;

\item $\vdash \;\subseteq \Con \; \times \; E$
is the \emph{enabling} relation;

\item $\Vdash \;\subseteq \Con \; \times \; E$ 
is the \emph{circular enabling} relation.

\end{itemize}
The relations $\vdash$ and $\Vdash$ are \emph{saturated}, i.e.  
\(
  \forall X \subseteq Y \subseteqfin E.\;\;
  X \circ e  \,\land\, CF(Y) 
  \implies 
  Y \circ e
\),
for $\circ \in \setenum{\vdash,\Vdash}$.
\end{definition}

A \CES\ is \emph{finite} when $E$ is finite;
it is \emph{conflict-free} when the relation $\#$ is empty.
We write
$a \vdash b$ for $\setenum{a} \vdash b$,
and $\vdash e$ for $\emptyset \vdash e$
(similar shorthands apply for $\Vdash$).

Intuitively, an enabling $X \vdash e$ models the fact that, 
if all the events in $X$ have happened, 
then $e$ is an obligation for participant $\princ{e}$;
such obligation may be discharged only by performing $e$, or 
any event in conflict with $e$.
For instance, an internal choice between $a$ and $b$
is modelled by a \CES\ with enablings $\vdash a$, $\vdash b$ 
and conflict $a \# b$.
After the choice (say, of $a$), the obligation $b$ is discharged.
The case of circular enablings $X \Vdash e$ is more complex: 
$e$ is an obligation if it is a \emph{prudent} event 
(see Def.~\ref{def:ces:prudence}).
Very roughly, $e$ is prudent when one can perform it ``on credit'',
and nevertheless be guaranteed that in all possible executions of the contract,
either the credit will be honoured (by doing the events in $X$), 
or the debtor will be culpable of a contract violation.
For instance, in the contract with enablings
$a \vdash b$ and $b \Vdash a$, 
the first enabling prescribes that $b$ can be done after $a$,
while the second enabling models
the fact that $a$ can be done on credit, 
on the guarantee that the other participant will be obliged to do $b$.
The event $a$ is prudent in the initial state,
because after doing it the other participant has the obligation
to perform~$b$ (not doing $b$ will result in a violation).

Besides the obligations, the other component of a contract is
a function $\payoffsym$ which specifies the objectives of each participant.
More precisely, $\payoffsym$ associates each participant {\pmv A} 
with a set of sequences in $E^{\infty}$ 
(the set of finite or infinite sequences on $E$), 
which represent those executions
where {\pmv A} has a positive payoff.

\begin{definition}[\bf Contract] \label{def:ces:contract}
A contract $\cname$ is a pair $\seq{\esname, \payoffsym}$, 
where $\esname$ is a \CES, and
$\payoffsym: \aname \rightarrow \powset{E^{\infty}}$ 
associates each participant with a set of traces.
\end{definition}

We interpret a contract 
as a nonzero-sum concurrent multi-player game.
The game involves the players in $\aname$ concurrently performing actions
in order to reach their objectives.
A \emph{play} of a contract $\cname$
is a conflict-free sequence $\sigma \in E^{\infty}$  without repetitions.
For $\sigma = \seq{e_0 \, e_1 \cdots} \in E^{\infty}$, 
we write $\mkset{\sigma}$ for the set of events in $\sigma$; 
we write  $\sigma_i$ for the subsequence $\seq{e_0 \cdots e_{i-1}}$.
If $\sigma = \seq{e_0 \cdots e_n}$, 
we write $\sigma \, e$ for $\seq{e_0 \cdots e_n \, e}$.
The empty sequence is denoted by~$\epsilon$.

Each play $\sigma = \seq{e_0 \cdots e_i \cdots}$ uniquely identifies a 
computation in the \CES~$\esname$.
This computation has the form 
\(
  (\emptyset,\emptyset) 
  \xrightarrow{e_0}
  (\mkset{\sigma_1},\mkcredit{\sigma_1}) 
  \cdots
  \xrightarrow{e_{i}}
  (\mkset{\sigma_{i+1}},\mkcredit{\sigma_{i+1}}) 
  \cdots
\).
The first element of each pair is the set of events occurred so far;
the second element is the least set of events done ``on credit'',
i.e.\ performed in the absence of a causal justification.
Formally, for all sequences $\eta = \seq{e_0 \; e_1 \cdots}$,
we define
\(
  \mkcredit{\eta} 
  = 
  \setcomp{e_i \in \mkset{\eta}}{\mkset{\eta_i} \not\vdash e_i \; \land \; \mkset{\eta}\not\Vdash e_i}
\).
Notice that $e \not\in \mkcredit{\eta}$ iff either $e$ is
$\vdash$-enabled by the past events $\mkset{\eta_i}$, or it is
$\Vdash$-enabled by the \emph{whole} play.

\begin{figure}[t]
\centering
\begin{tabular}{ccccc}

\begin{tikzpicture}[scale=1.2]
\draw [->] (0,0)  node [above] {$a$} --(0.95,0) node [above] {$b$};
\draw [->] (-0.5,0)   --(-0.05,0);
\draw [fill=black] (0,0) circle (0.05);
\draw [fill=black] (1,0) circle (0.05);
\node [below] at (0.5, -0.4) {$\esname_{1}$};
\end{tikzpicture}

& \hspace{15pt}

\begin{tikzpicture}[scale=1.2]
\draw [->, rounded corners] (2,0) node [above] {$a\;\;$} -- (2.5,0.4) -- (2.98,0.01) node [above] {$\;\;b$};
\draw [->, rounded corners] (3,0) -- (2.5, -0.4) -- (2.01,-0.01);
\draw [fill=black] (1.95,0) circle (0.05);
\draw [fill=black] (+3.05,0) circle (0.05);
\node [below] at (2.5, -0.4) {$\esname_{2}$};
\end{tikzpicture}

& \hspace{15pt}

\begin{tikzpicture}[scale=1.2]
\draw [->, rounded corners] (2,0) node [above] {$a\;\;$} -- (2.5,0.4) -- (2.98,0.01) node [above] {$\;\;b$};
\draw [->>, rounded corners] (3,0) -- (2.5, -0.4) -- (2.01,-0.01);
\draw [fill=black] (1.95,0) circle (0.05);
\draw [fill=black] (+3.05,0) circle (0.05);
\node [below] at (2.5, -0.4) {$\esname_{3}$};
\end{tikzpicture}

& \hspace{15pt}

\begin{tikzpicture}[scale=1.2]
\draw [->>, rounded corners] (2,0) node [above] {$a\;\;$} -- (2.5,0.4) -- (2.98,0.01) node [above] {$\;\;b$};
\draw [->>, rounded corners] (3,0) -- (2.5, -0.4) -- (2.01,-0.01);
\draw [fill=black] (1.95,0) circle (0.05);
\draw [fill=black] (+3.05,0) circle (0.05);
\node [below] at (2.5, -0.4) {$\esname_{4}$};
\end{tikzpicture}

& \hspace{15pt}

\begin{tikzpicture}[scale=1]
\draw [->] (4,0.5) node [left] {$a$\,} --(5,0) node [right] {$c$};
\draw [->] (4,0.5) --(5,1) node [right] {$b$};
\draw [fill=black] (3.95,0.5) circle (0.05);
\draw [fill=black] (+5.05,0) circle (0.05);
\draw [fill=black] (+5.05,1) circle (0.05);
\draw [decorate,decoration={snake,amplitude=.3mm,segment length=2mm,post length=1mm}] (5.05,0) --(5.05,1) ;
\draw [->, rounded corners] (5,-0.05) -- ( 4.25,-0.05) -- (3.95,0.45);
\draw [->>, rounded corners] (5,1.05) -- (4.25, 1.05) -- (3.95,0.55);
\node [below] at (4.5, -0) {$\esname_{5}$};
\end{tikzpicture}

\end{tabular}
\caption{
Graphical representation of \CES.
An hyperedge from a set of nodes $X$ to $e$ 
denotes an enabling $X \circ e$, where 
$\circ = \; \vdash$ if the edge has a single arrow, and 
$\circ = \; \Vdash$ if it has a double arrow.
A conflict $a \# b$ is represented by a wavy line
between $a$ and $b$. 
}
\label{fig:ces}
\end{figure}

\begin{example} \label{ex:ces:credits}
Consider the \CES\ in Fig.~\ref{fig:ces}.
The maximal plays of $\esname_1$--$\esname_4$ are $\seq{ab}$, $\seq{ba}$, 
for which we have the following computations:
\begin{description}

\item[$\esname_{1}:$]
\(
  (\emptyset,\emptyset) \xrightarrow{a} (\setenum{a},\emptyset) \xrightarrow{b} (\setenum{a,b},\emptyset)
\),
\hspace{37pt}
\(
  (\emptyset,\emptyset) \xrightarrow{b} (\setenum{b},\setenum{b}) \xrightarrow{a} (\setenum{a,b},\setenum{b})
\).

\item[$\esname_{2}:$] 
\(
  (\emptyset,\emptyset) \xrightarrow{a} (\setenum{a},\setenum{a}) \xrightarrow{b} (\setenum{a,b},\setenum{a})
\),
\hspace{17pt}
\(
  (\emptyset,\emptyset) \xrightarrow{b} (\setenum{b},\setenum{b}) \xrightarrow{a} (\setenum{a,b},\setenum{b})
\).

\item[$\esname_{3}:$]  
\(
  (\emptyset,\emptyset) \xrightarrow{a} (\setenum{a},\setenum{a}) \xrightarrow{b} (\setenum{a,b},\emptyset)
\),
\hspace{27pt}
\(
  (\emptyset,\emptyset) \xrightarrow{b} (\setenum{b},\setenum{b}) \xrightarrow{a} (\setenum{a,b},\setenum{b})
\).

\item[$\esname_{4}:$]  
\(
  (\emptyset,\emptyset) \xrightarrow{a} (\setenum{a},\setenum{a}) \xrightarrow{b} (\setenum{a,b},\emptyset)
\),
\hspace{27pt}
\(
  (\emptyset,\emptyset) \xrightarrow{b} (\setenum{b},\setenum{b}) \xrightarrow{a} (\setenum{a,b},\emptyset)
\).

\end{description}
The maximal plays of $\esname_5$ are 
$\seq{ab}$, $\seq{ba}$, $\seq{ac}$, $\seq{ca}$.
For $\seq{ab}$, $\seq{ba}$, the computations are as those of $\esname_3$,
while for $\seq{ac}$, $\seq{ca}$ the computations are 
as those of $\esname_2$ (with $c$ in place of $b$).
\end{example}

A \emph{strategy} $\Sigma$ for {\pmv A} is a function which associates 
to each finite play $\sigma$ a set of events of {\pmv A} 
such that
if $e \in \Sigma(\sigma)$ then $\sigma e$ is still a play.
A play $\sigma = \seq{e_0 \, e_1 \cdots}$ \emph{conforms} to a strategy $\Sigma$ 
for {\pmv A} if, for all $i \geq 0$, if $e_i \in \invprinc{\pmv A}$, 
then $e_i \in \Sigma(\sigma_i)$.
A play is \emph{fair} w.r.t.\ a strategy $\Sigma$ when 
there are no events in $\sigma$ which are perpetually enabled by $\Sigma$.

Before setting up the crucial notions of fair play and prudent events,
we provide some underlying intuitions.
The definition of prudent strategies and of innocent participants
is mutually coinductive.
A participant {\pmv A} is considered \emph{innocent} 
in a play $\sigma$ when she has done all her prudent events in $\sigma$
(otherwise {\pmv A} is \emph{culpable}).
Hence, if a strategy tells {\pmv A} to do all her prudent events,
then in all \emph{fair} plays these events must either become imprudent,
or be fired.
Given a finite play $\sigma$ of past events, 
an event $e$ is said \emph{prudent} in $\sigma$ whenever
there exists a prudent strategy $\Sigma$ 
which prescribes to do $e$ in $\sigma$.
A strategy for {\pmv A} with past~$\sigma$ (namely, conforming to $\sigma$) 
is prudent whenever, in all fair extensions of $\sigma$
where all other participants are innocent,
the events performed on credit by {\pmv A} are eventually honoured;
at most, the credits coming from the past $\sigma$ will be left.
Notice that we neglect those \emph{unfair} plays
where an action permanently enabled is not eventually performed.
Indeed, an unfair scheduler could perpetually prevent
an honest participant 
from performing a promised action.

\begin{definition}[Fair play] \label{def:ces:fair-play}
A play $\sigma = \seq{e_0 \, e_1 \cdots}$ 
is \emph{fair} w.r.t.\ strategy $\Sigma$ iff:
\[
  \forall i \leq |\sigma|.\;
  \big(
  \forall j : i \leq j \leq |\sigma|.\;
  e \in \Sigma(\sigma_j)
  \big)
  \implies 
  \exists h \geq i.\;
  e_h = e
\]
\end{definition}

\begin{definition}[Prudence] \label{def:ces:prudence}
A strategy $\Sigma$ for {\pmv A} with past $\sigma$ is \emph{prudent} if,
for all fair plays $\sigma'$ extending $\sigma$,
conforming to~$\Sigma$, and
where all ${\pmv B} \neq {\pmv A}$ are innocent,
\[
  \exists k > |\sigma|. \;
  \; \mkcredit{\sigma'_k} \cap \invprinc{\pmv A} 
  \,\subseteq\, 
  \mkcredit{\sigma}
\]
An event $e$ is \emph{prudent} in $\sigma$ if
there exists a prudent strategy $\Sigma$ with past $\sigma$
such that $e \in \Sigma(\sigma)$.

\smallskip\noindent
A participant {\pmv A} 
is \emph{innocent} in $\sigma = \seq{e_0 \, e_1 \cdots}$ iff:
\[
  \forall e \in \invprinc{\pmv A}.\;
  \forall i \geq 0.\;
  \exists j \geq i.\;
  e \text{ is imprudent in } \sigma_j
\]
\end{definition}

Notice that the empty strategy is trivially prudent.

\begin{example} \label{ex:ces:prudence}
Consider the obligations modelled by the five \CES\ in Fig.~\ref{fig:ces},
where $\princ{a} = {\pmv A}$ and $\princ{b} = \princ{c} = {\pmv B}$:
\begin{itemize}

\item in $\esname_{1}$, the only prudent event in the empty play is $a$,
which is enabled by $\emptyset$, and the only culpable participant is {\pmv A}.
In $\seq{a}$, $b$ becomes prudent, and {\pmv B} becomes culpable.
In $\seq{ab}$ no event is prudent and no participant is culpable.

\item in $\esname_2$, there are no prudent events in $\epsilon$.
Instead, event $a$ is prudent in $\seq{b}$, while $b$ is prudent in $\seq{a}$:
this is coherent with the fact that the prudence of an event
does not depend on the assumption that all the events done 
in the past were prudent.
In $\seq{ab}$ and $\seq{ba}$ no events are prudent.

\item in $\esname_3$, event $a$ is prudent in $\epsilon$:
indeed, the only fair play $a \eta$ where {\pmv B}
is innocent is $\seq{ab}$, where $\mkcredit{ab} = \emptyset$.
Instead, $b$ is \emph{not} prudent in $\epsilon$,
because $b \in \mkcredit{b\eta}$ for all $\eta$.
Event $b$ is prudent in $\seq{a}$.

\item in $\esname_4$, both $a$ and $b$ are prudent in $\epsilon$.

\item in $\esname_{5}$,
$a$ is \emph{not} prudent in $\epsilon$,
because if {\pmv B} chooses to do $c$, 
then the credit $a$ can no longer be honoured.
Actually, no events are prudent in $\epsilon$, 
while both $b$ and $c$ are prudent in $\seq{a}$,
and $a$ is prudent in both $\seq{b}$ and $\seq{c}$.

\end{itemize}
\end{example}

We now define when a participant \emph{wins} in a play.
If {\pmv A} is culpable, then she loses.
If {\pmv A} is innocent, but some other participant is culpable,
then {\pmv A} wins.
Otherwise, if all participants are innocent, then 
{\pmv A} wins if she has a positive payoff in the play,
and the play is ``credit-free''.

\begin{definition}[\bf Winning play] \label{def:ces:win}
Define the function $\winsym: \aname \rightarrow \powset{E^{\infty}}$ as follows:
\begin{align*}
  \win{\pmv A}{} 
  \; = \;
  & \setcomp{\sigma \in \payoff{\pmv A}{}}
  {\text{{\pmv A} credit-free in $\sigma$, and all participants are innocent in $\sigma$}} 
  \; \cup \\
  & \setcomp{\sigma}{\text{{\pmv A} innocent in $\sigma$, and some ${\pmv B} \neq {\pmv A}$ is culpable in $\sigma$}}
\end{align*}
where {\pmv A} is \emph{credit-free} in $\sigma$ iff:
\(\;
  \forall e \in \princsym^{-1}(\pmv A).\;
  \forall i \geq 0.\;
  \exists j \geq i.\;
  e \not\in \mkcredit{\sigma_{j}}
\).
\end{definition}

\begin{example}
Notice that innocence and credit-freeness are distinct notions.
For instance, for the contract induced by the \CES\ $\esname_3$ in
Fig.~\ref{fig:ces},
assuming $\princ{a} = {\pmv A}$ and $\princ{b} = \princ{c} = {\pmv B}$,
we have that
in $\sigma = \epsilon$, 
{\pmv A} is credit-free, but not innocent (because a is prudent in $\epsilon$),
in $\sigma = \seq{a}$, {\pmv A} is innocent, but not credit-free 
(because $\mkcredit{\seq{a}} = \setenum{a}$), and
in $\sigma = \seq{ab}$, {\pmv A} is innocent and credit-free.
\end{example}

A key property of contracts is that of \emph{agreement}.
Intuitively, when {\pmv A} agrees on a contract $\cname$,
then she can safely initiate an interaction with the other participants, 
and be guaranteed that the interaction will not ``go wrong''
--- even in the presence of attackers.
This does not mean that {\pmv A} will always succeed in all interactions:
in case {\pmv B} is dishonest, we do not assume that an external authority
will disposses {\pmv B} of $b$ and give it to {\pmv A}.
Participant {\pmv A} will agree on a contract where
she reaches her goals, or she can blame another participant
for a contract violation.
In real-world applications, a judge may provide compensations to {\pmv A},
or impose a punishment to the culpable participant.

We now define when a participant \emph{agrees} on a contract.
We say that $\Sigma$ is \emph{winning} for {\pmv A} iff {\pmv A} wins 
in every fair play which conforms to $\Sigma$.
Intuitively, {\pmv A} is happy to participate in an interaction 
regulated by contract $\cname$ when she has a strategy $\Sigma$
which allows her to win in all fair plays conforming to $\Sigma$.

\begin{definition}[\bf Agreement] \label{def:ces:agreement}
A participant {\pmv A} \emph{agrees} on a contract $\cname$ whenever
{\pmv A} has a winning strategy in $\cname$.
A contract $\cname$ \emph{admits an agreement} whenever all the involved participants 
agree on $\cname$.
\end{definition}

\begin{example} \label{ex:ces:agreement}
Consider the contracts $\cname_i$ where 
the obligations are specified by $\esname_i$ in Fig.~\ref{ex:ces:prudence},
and let the goals of {\pmv A} and {\pmv B} be as follows:
{\pmv A} is happy when she obtains $b$ 
(i.e.\ $\payoffsym{\pmv A} = \setcomp{\sigma}{b \in \mkset{\sigma}}$),
while {\pmv B} is happy when he obtains $a$
($\payoffsym{\pmv B} = \setcomp{\sigma}{a \in \mkset{\sigma}}$).
\begin{itemize}

\item $\cname_{1}$ admits an agreement.
The winning strategies for {\pmv A} and {\pmv B} are, respectively, 
\[
  \Sigma_{\pmv A}(\sigma) = \begin{cases}
    \setenum{a} & \text{if $a \not\in \mkset{\sigma}$} \\
    \emptyset & \text{otherwise}
  \end{cases}
  \hspace{40pt}
  \Sigma_{\pmv B}(\sigma) = \begin{cases}
    \setenum{b} & \text{if $a \in \mkset{\sigma}$ and $b \not\in \mkset{\sigma}$} \\
    \emptyset & \text{otherwise}
  \end{cases}
\]
Roughly, 
the only fair play conforming to $\Sigma_{\pmv A}$ and $\Sigma_{\pmv B}$
where both {\pmv A} and {\pmv B} are innocent 
is $\sigma = \seq{a b}$.
We have that {\pmv A} and {\pmv B} win in $\sigma$, because
both participants are credit-free in $\sigma$
(see Ex.~\ref{ex:ces:credits}),
and $\sigma \in \payoff{\pmv A}{} \cap \payoff{\pmv B}{}$.

\item $\cname_2$ does not admit an agreement.
Indeed, there are no prudent events in $\epsilon$,
hence both {\pmv A} and {\pmv B} are innocent in $\epsilon$.
If no participant takes the first step, then nobody reaches her goals.
If a participant takes the first step, then the resulting trace
is not credit-free.
Thus, no winning strategy exists.

\item $\cname_3$ admits an agreement.
The winning strategies are as for $\cname_1$ above:
{\pmv A} first does $a$, then {\pmv B} does $b$.
While $\cname_1$ and $\cname_3$ are identical from the point
of view of agreements, they differ in that
$\cname_3$ \emph{protects} {\pmv A}, while $\cname_1$ does not.
Intuitively, the enabling $\vdash a$ in $\cname_1$
models an obligation for {\pmv A} also in those contexts 
where no agreement exists,
while $b \Vdash a$ only forces {\pmv A} to do $a$
when $b$ is guaranteed.

\item $\cname_4$ admits an agreement.
In this case the winning strategies for {\pmv A} and {\pmv B} are:
\[
  \Sigma_{\pmv A}(\sigma) = \begin{cases}
    \setenum{a} & \text{if $a \not\in \mkset{\sigma}$} \\
    \emptyset & \text{otherwise}
  \end{cases}
  \hspace{40pt}
  \Sigma_{\pmv B}(\sigma) = \begin{cases}
    \setenum{b} & \text{if $b \not\in \mkset{\sigma}$} \\
    \emptyset & \text{otherwise}
  \end{cases}
\]
That is, a participant must be ready to do her action without waiting
for the other participant to make the first step.

\item $\cname_{5}$ does not admit an agreement.
Since no events are prudent in $\epsilon$,
both participants are innocent in $\epsilon$,
but if they cannot reach their goals by doing nothing.
If {\pmv A} does $a$, then {\pmv B} can choose to do $c$.
This makes {\pmv B} innocent (and winning), but then {\pmv A}
loses, because not credit-free in $\seq{a c}$.

\end{itemize}
\end{example}

  \subsection{Propositional Contract Logic} \label{pcl}

We briefly review Propositional Contract Logic 
(\pcl~\cite{BZ10lics}),
\pcl\ extends intuitionistic propositional logic IPC
with the connective~$\coimp$, called \emph{contractual implication}.
We assume that the atoms of \pcl\ are the events in $E$.
The formulae of \pcl\ are defined as follows:
\[
  p,q \;\; ::= \;\;
  \bot \;\mid\; \top \;\mid\; \atom{a} \;\mid\; \lnot p \;\mid\; 
  p \lor q \;\mid\; p \land q \;\mid\; p \imp q \;\mid\; p \coimp q
\]

A proof system for \pcl\ is defined in~\cite{BZ10lics} in terms 
of Gentzen-style rules (Fig.~\ref{fig:pcl:gentzen}),
which extend those of IPC.
In all the rules, $\Delta$ is a set of \pcl\ formulae.
Decidability of \pcl\ has been established in~\cite{BZ10lics} 
by proving that the Gentzen-style proof system of \pcl\
enjoys cut elimination and the subformula property.

\begin{figure}[t]
\[
\begin{array}{c}
  \irule
  {\Delta \;\vdash\; q}
  {\Delta \;\vdash\; p \coimp q}
  \;\nrule{(Zero)} 
  \qquad
  \irule
  {\begin{array}{l}
      \Delta,\ p \coimp q,\ c \;\vdash\; p \\
      \Delta,\ p \coimp q,\ q \;\vdash\; c \coimp d
  \end{array}}
  {\Delta,\ p \coimp q \;\vdash\; c \coimp d}
  \;\nrule{(Lax)} 
  \qquad
  \irule
  {\begin{array}{c}
      \Delta,\ p \coimp q,\ r \;\vdash\; p \\
  \Delta,\ p \coimp q,\ q \;\vdash\; r
  \end{array}
  }
  {\Delta,\ p \coimp q \;\vdash\; r}
  \;\nrule{(Fix)} 
\end{array}
\]
\caption{Sequent calculus for \pcl\ (rules for $\coimp$; the full set of rules is in~\iftoggle{proofs}{Fig.~\ref{fig:pcl:gentzen-full}}{\cite{ces-pcl-long}}).}
\label{fig:pcl:gentzen}
\end{figure}

In this paper we shall mainly consider the Horn fragment of \pcl, 
which comprises atoms, conjunctions, and non-nested 
$\imp$/$\coimp$ implications.
Let $\alpha,\beta$ range over conjunctions of atoms.
A \emph{Horn \pcl\ theory} is a \fxnote{finite} set of clauses of the form
$\alpha \imp a$ or $\alpha \coimp a$.
The clause $a$ is a shorthand for $\top \imp a$.
We shall denote with $\mkset{\alpha}$ the set of atoms in $\alpha$.

\section{Proof traces in \pcl} \label{proof-traces}

In this section we introduce the notion of \emph{proof traces},
namely the sequences of atoms respecting the order imposed by proofs in
\pcl.
To do that, we first define a natural deduction system for \pcl,
which extends that of IPC with the rules in Fig.~\ref{fig:pcl:nd}.
Provable formulae are contractually implied,
according to rule~\nrule{($\coimp$I1)}.
Rule~\nrule{($\coimp$I2)} provides $\coimp$ with the same weakening properties
of $\imp$.
The crucial rule is \nrule{($\coimp$E)}, which allows for the elimination of $\coimp$.
Compared to the rule for elimination of $\imp$ in IPC, the only difference is that
in the context used to deduce the antecedent $p$, 
rule \nrule{($\coimp$E)} also allows for using as hypothesis the consequence~$q$.

\begin{example} \label{ex:pcl:proof}
Let $\Delta = a \imp b, b \coimp a$. 
A proof of $\Delta \vdash a$ in natural deduction is:
\[
  \irule
  {\Delta \vdash b \coimp a \quad
  \irule
    {\Delta \vdash a \imp b \quad 
      \Delta,a \vdash a} 
    {\Delta,a \vdash b} 
    \nrule{($\imp$E)}}
  {\Delta \vdash a}
  \nrule{($\coimp$E)}  
\]
\end{example}

\begin{figure}[t]
\[
  \irule
  {\Delta \vdash q}
  {\Delta \vdash p \coimp q}
  \;\nrule{($\coimp$I1)}
  \quad\;\;
  \irule
  {\begin{array}{c} \\ \Delta \vdash p' \coimp q' \end{array} \quad
   \begin{array}{l}
     \Delta,p \vdash p' \\
     \Delta,q' \vdash p \coimp q
   \end{array}}
 {\Delta \vdash p \coimp q}
  \;\nrule{($\coimp$I2)}
  \quad\;\;
  \irule
  {\Delta \vdash p \coimp q \quad \Delta, q \vdash p}
  {\Delta \vdash q}
  \;\nrule{($\coimp$E)}
\]
\vspace{-10pt}
\caption{Natural deduction for \pcl\ (rules for $\coimp$; 
the full set of rules is in~\iftoggle{proofs}{Fig.~\ref{fig:pcl:nd-full}}{\cite{ces-pcl-long}}).}
\label{fig:pcl:nd}
\end{figure}

The natural deduction system of Fig.~\ref{fig:pcl:nd} 
is equivalent to the Gentzen calculus of~\cite{BZ10lics}.

\begin{restatable}{mythm}{ndgentzen}
\label{th:pcl:nd-gentzen}
Thre exists a proof $\pi$ of $\Delta \vdash p$  in natural deduction iff 
there exists a proof $\pi^{*}$ of $\Delta \vdash p$  in the sequent calculus 
of~\cite{BZ10lics}.
\end{restatable}

For proving atoms (or their conjunctions) in Horn \pcl\ theories, 
a strict subset of the natural deduction rules suffices.

\begin{restatable}{mylem}{ndhorn}
\label{lem:pcl:nd-horn}
Let $\Delta$ be a Horn \pcl\ theory.
If $\Delta \vdash \alpha$ in natural deduction,
then a proof of $\Delta \vdash \alpha$ exists which uses only the rules
\nrule{(Id)}, \nrule{($\land$I)}, \nrule{($\land$E1)},\nrule{($\land$E2)},
\nrule{($\imp$E)}, and \nrule{($\coimp$E)}.
\end{restatable}

A key observation is that each proof in Horn \pcl\ induces
a set of atom orderings which are compatible with the proof.
Each of these orderings is associated with a sequence of atoms,
called \emph{proof trace}.
To give some intuition, consider the elimination rule for $\imp$:
\[
  \irule
  {\Delta \vdash \alpha \imp a \qquad \Delta \vdash \alpha}
  {\Delta \vdash a}
  \;\nrule{($\imp$E)}
\]
The rule requires a proof of all the atoms in $\alpha$ in order to 
construct a proof of~$a$.
Accordingly, if $\sigma$ is a proof trace of $\Delta$,
then $\sigma a$ if a proof trace of $\Delta$.

Consider now the elimination rule for $\coimp$
\[
  \irule
  {\Delta \vdash \alpha \coimp a \qquad \Delta,a \vdash \alpha}
  {\Delta \vdash a}
  \;\nrule{($\coimp$E)}
\]
Here, the intuition is that $\alpha$ needs not necessarily be proved before $a$:
it suffices to prove $\alpha$ by taking $a$ as hypothesis.
Assuming that $\sigma$ is a proof trace of $\Delta,a$,
the proof traces of $\Delta$ include all the interleavings
between $\sigma$ and $a$.

\begin{definition}[Proof traces]
\label{def:pcl:proof-trace}
For a Horn \pcl\ theory $\Delta$, 
we define the set of sequences of atoms $\sem{\Delta}$
by the rules in Fig.~\ref{fig:pcl:proof-traces}.
For $\sigma,\eta \in E^*$, we denote with
$\sigma \eta$ the concatenation of $\sigma$ and $\eta$,
and with $\sigma \mid \eta$ the set of interleavings 
of $\sigma$ and $\eta$.
We assume that both operators remove duplicates from the right,
e.g.\ $aba \mid ca = ab \mid ca = \setenum{abc,acb,cab}$.
We call each $\sigma \in \sem{\Delta}$ a \emph{proof trace} of $\Delta$.
\end{definition}

\begin{figure}[t]
\[
\begin{array}{c}
  \irule
  {}
  {\epsilon \in \sem{\Delta}}
  \,\nrule{($\epsilon$)}
  \hspace{12pt}
  \irule
  {\alpha \imp a \in \Delta \quad
   \sigma \in \sem{\Delta} \quad 
   \mkset{\alpha} \subseteq \mkset{\sigma}}
  {\sigma \, a \in \sem{\Delta}}
  \,\nrule{($\imp$)}
  \hspace{12pt}
  \irule
  {\alpha \coimp a \in \Delta \quad 
   \sigma \in \sem{\Delta, a} \quad
   \mkset{\alpha} \subseteq \mkset{\sigma}}
  {\sigma \mid a \,\subseteq\, \sem{\Delta}}
  \,\nrule{($\coimp$)}
\end{array}
\]
\vspace{-10pt}
\caption{Proof traces of Horn \pcl.}
\label{fig:pcl:proof-traces}
\end{figure}

\begin{example} \label{ex:pcl:proof-trace}
Consider the following Horn \pcl\ theories
(recall that $a \eqdef \top \imp a$):
\[
\begin{array}{ll}
  \Delta_1 = \setenum{a \imp b,\, a}
  \hspace{40pt}
  &
  \Delta_2 = \setenum{a \imp b,\, b \imp a}
  \\
  \Delta_3 = \setenum{a \imp b,\, b \coimp a}
  &
  \Delta_4 = \setenum{a \coimp b,\, b \coimp a}
\end{array}
\]
(notice the resemblance with the \CES\ 
$\esname_1$--$\,\esname_4$ in Fig.~\ref{fig:ces}).
By Def.~\ref{def:pcl:proof-trace}, we have:
\[
\begin{array}{ll}
  \sem{\Delta_1} = \setenum{\epsilon,\, a,\, ab}
  \hspace{40pt}
  &
  \sem{\Delta_2} = \setenum{\epsilon}
  \\
  \sem{\Delta_3} = \setenum{\epsilon,\, ab}
  &
  \sem{\Delta_4} = \setenum{\epsilon,\, ab, \, ba}
\end{array}
\]
For instance, we deduce $ab \in \sem{\Delta_3}$ through the following derivation:
\vspace{-5pt}
\begin{small}
\[
  \irule
  {\begin{array}{c} \\ b \coimp a \in \Delta_3 \end{array} \;
  \irule
    {\begin{array}{c} \\ a \imp b \in \Delta_3,a \end{array} \; 
     \irule{
       \begin{array}{c} \\ \top \imp a \in \Delta_3,a \end{array}
       \quad
       \irule{}{\epsilon \in \sem{\Delta_3,a}} \nrule{($\epsilon$)}
     }
     {a \in \sem{\Delta_3,a}} \nrule{($\imp$)} \;
     \begin{array}{c} \\ \mkset{a} \subseteq \mkset{a} \end{array}}
    {a b \in \sem{\Delta_3,a}} \nrule{($\imp$)} \;
     \begin{array}{c} \\ \mkset{b} \subseteq \mkset{ab} \end{array}}
  {ab = ab \mid a \in \sem{\Delta_3}} \nrule{($\coimp$)}
\]
\end{small}
Notice that $ba \not\in \sem{\Delta_3}$:
indeed, to derive any non-empty $\alpha$ from $\Delta_3$
one needs to use both $a \imp b$ and $b \coimp a$,
hence all non-empty proof traces must contain
both $a$ and $b$;
since $b$ does not occur at the right of a contractual implication,
it cannot be interleaved; thus, $ba$ is not derivable.
\fxnote{check}
\end{example}

We now define, starting from a set $X$ of atoms, 
which atoms may be proved immediately after,
following some proof trace. 
We call these atoms \emph{urgent}, 
and we denote with $\curgent[\Delta]{X}$
the set of urgent atoms in $X$.
For instance, with $\Delta_1$ in Ex.~\ref{ex:pcl:proof-trace},
we have $\curgent[\Delta_1]{\emptyset} = \setenum{a}$,
$\curgent[\Delta_1]{\setenum{a}} = \setenum{b}$, 
$\curgent[\Delta_1]{\setenum{b}} = \setenum{a}$, and 
$\curgent[\Delta_1]{\setenum{a,b}} = \emptyset$.

\begin{definition} \label{def:pcl:urgent}
For a set $X \subseteq E$ and a Horn \pcl\ theory $\Delta$, we define
$\curgent[\Delta]{X}$ as:
\[
  \curgent[\Delta]{X} 
  \;\; = \;\;
  \setcomp{a \not\in X}
  {
  \exists \sigma, \sigma'. \;\;
  \mkset{\sigma} = X  
  \;\land\;
  \sigma \, a \, \sigma' \in \sem{\Delta, X}
  }
\]
\end{definition}

Theorem~\ref{th:pcl:enc-u} below characterizes urgent atoms
in terms of provability.
This is obtained by a suitable rewriting of Horn \pcl\ theories, 
which separates the urgent atoms from the provable ones.

Technically, in Def.~\ref{def:pcl:enc-u} we introduce
an endomorphism $\enc[\urgent]{\cdot}$ of Horn \pcl\ theories.
Let $\star \in \setenum{!,R,U}$.
We assume three injections $\star: \atoms \rightarrow \atoms$, 
such that $\amod[!]{\atoms}$, $\amod[R]{\atoms}$ and $\amod[U]{\atoms}$ 
are pairwise disjoint. 
For a set of atoms $X \subseteq \atoms$,
we denote with $\amod[\star]{X}$ the theory $\setcomp{\amod[\star]{e}}{e \in X}$.
Intuitively, the atoms of the form $\amod[!]{a}$ correspond to
actions already happened in the past,
the atoms $\amod[U]{a}$ comprise the urgent actions,
while the atoms $\amod[R]{a}$ are those actions which can be eventually 
reached by performing the urgent ones.

Below, we denote with $\primes{\Delta}$ the set of all atoms in $\Delta$.
We assume that $\primes{\Delta} \cap\, \amod[\star]{\atoms} = \emptyset$,
and that $a$ stands for an atom not in $\amod[\star]{E}$.
For a set $X \subseteq \amod[!]{E} \cup \amod[R]{E} \cup \amod[U]{E}$,
we define the projection
\(
  \unmod[\star]{X} = \setcomp{e \in E}{\amod[\star]{e} \in X}
\).
When $\alpha = a_1 \land \cdots \land a_n$, we write
$\amod[\star]{\alpha} = \amod[\star]{a_1} \land \cdots \land \amod[\star]{a_n}$.
When $n = 0$, $\amod[\star]{\alpha} = \top$.

\begin{definition} \label{def:pcl:enc-u} 
The endomorphism $\enc[\urgent]{\cdot}$ of Horn \pcl\ theories is defined as:
\begin{align*}
  \enc[\urgent]{\Delta, \alpha \circ \atom{a}}
  \; & = \; \textstyle
  \enc[\urgent]{\Delta},\;
  {\enc[\urgent]{\alpha \circ \atom{a}}},\;
  \Omega(\primes{\alpha \circ \atom{a}})
  && \text{for $\circ \in \setenum{\imp,\coimp}$}
  \\
  \Omega(X) 
  \; & = \; 
  \textstyle
  \setcomp{\amod[!]{\atom{a}} \imp \amod[U]{\atom{a}}}{\atom{a} \in X} 
  \;\cup\; 
  \setcomp{\amod[U]{\atom{a}} \imp \amod[R]{\atom{a}}}{\atom{a} \in X}
  \\
  \enc[\urgent]{\alpha \imp \atom{a}}
  \; & = \;
  \setenum{\amod[!]{\alpha} \imp \amod[U]{\atom{a}},\; \amod[R]{\alpha} \imp \amod[R]{\atom{a}}}
  \\
  \enc[\urgent]{\alpha \coimp \atom{a}}
  \; & = \;
  \setenum{\amod[R]{\alpha} \coimp \amod[U]{\atom{a}}}
\end{align*}
\end{definition}

The encoding of an implication $\alpha \imp a$ contains
$\amod[!]{\alpha} \imp \amod[U]{a}$, 
meaning that $a$ becomes urgent when its preconditions $\alpha$ have been done,
and $\amod[R]{\alpha} \imp \amod[R]{a}$, meaning that
$a$ is reachable whenever its preconditions are such.
The encoding of a contractual implication $\alpha \coimp a$ contains
$\amod[R]{\alpha} \coimp \amod[U]{a}$, meaning that
$a$ is urgent when its preconditions are guaranteed to be reachable.

\begin{example}
For the \pcl\ theory $\Delta_3 = \setenum{a \imp b,\; b \coimp a}$
in Ex.~\ref{ex:pcl:proof-trace}, we have:
\begin{align*}
  \enc[\urgent]{\Delta_3} \; = \;\;
  \{ &
  \amod[!]{\atom a} \imp \amod[U]{\atom b},
  \;\;
  \amod[R]{\atom a} \imp \amod[R]{\atom b},
  \;\;
  \amod[R]{\atom b} \coimp \amod[U]{\atom a},
  \;\;
  \\
  & \amod[!]{\atom a} \imp \amod[U]{\atom a},
  \;\;
  \amod[!]{\atom b} \imp \amod[U]{\atom b},
  \;\;
  \amod[U]{\atom a} \imp \amod[R]{\atom a},
  \;\;
  \amod[U]{\atom b} \imp \amod[R]{\atom b}
  \}
\end{align*}
We have that 
$\enc[\urgent]{\Delta_3} \vdash \amod[U]{\atom a}$ and
$\enc[\urgent]{\Delta_3} \not\vdash \amod[U]{\atom b}$;
also, $\enc[\urgent]{\Delta_3},\amod[!]{\atom a} \vdash \amod[U]{\atom b}$.
Notice that if the clause $b \coimp a$ were mapped by $\enc[\urgent]{\,}$ 
to $\amod[R]{b} \imp \amod[U]{a}$ (without contractual implication), 
then no atoms would have been provable in $\enc[\urgent]{\Delta_3}$.
\end{example}

The following lemma states that the atoms $a$ for which $\amod[R]{a}$
is derivable from $\enc[\urgent]{\Delta}$
are exactly those atoms which occur in some proof trace of $\Delta$.

\begin{restatable}{mylem}{proofreach}
\label{lem:pcl:proof-reach}
\(
  a \in \bigcup \mkset{\sem{\Delta}}
  \iff
  \enc[\urgent]{\Delta} \vdash \amod[R]{a}  
\)
\end{restatable}

The following lemma relates proof traces with 
urgent atoms derivable from $\enc[\urgent]{\Delta}$.
The $(\Leftarrow)$ direction states that (any prefix of)
a proof trace is made by urgent atoms in sequence.
The $(\Rightarrow)$ direction states that a sequence of 
urgent atoms can be extended to a proof trace.

\begin{restatable}{mylem}{proofurgent}
\label{lem:pcl:proof-urgent}
Let $\sigma = \seq{e_0 \cdots e_n}$. 
Then, 
\[
  \forall i \in 0..n.\;
  \enc[\urgent]{\Delta},\, \amod[!]{\mkset{\sigma_i}} 
  \vdash 
  \amod[U]{e_i}
  \;\;\iff\;\;
  \exists \eta.\; \sigma\eta \in \sem{\Delta} 
\]
\end{restatable}

The main result about the endomorphism $\enc[\urgent]{\,}$ follows.
Given a Horn \pcl\ theory $\Delta$,
an atom $\atom{a}$ is urgent in $\Delta$ iff
$\amod[U]{\atom{a}}$ is provable in $\enc[\urgent]{\Delta}$.

\begin{theorem} \label{th:pcl:enc-u}
For all Horn \pcl\ theories $\Delta$, 
and for all $a \not\in X \subseteq \atoms$:
\[
  \atom{a} \in \curgent[\Delta]{X}
  \;\; \iff \;\;
  \enc[\urgent]{\Delta}, \amod[!]{X} \vdash \amod[U]{\atom{a}} 
\]
\end{theorem}
\begin{proof}
$(\Rightarrow)$
Assume that $\atom{a} \in \curgent[\Delta]{X}$.
By Def.~\ref{def:pcl:urgent}, there exist $\sigma,\sigma'$ such that
$\mkset{\sigma} = X$ and $\sigma a \sigma' \in \sem{\Delta, X}$.
By Lemma~\ref{lem:pcl:proof-urgent}, we have 
\(
  \enc[\urgent]{\Delta, X}, \amod[!]{X} \vdash \amod[U]{a}
\).
The thesis
\(
  \enc[\urgent]{\Delta}, \amod[!]{X} \vdash \amod[U]{a}
\)
follows because $\enc[\urgent]{X} = \amod[!]{\top} \imp \amod[U]{X}$ and
$\amod[!]{X}$ implies $\amod[U]{X}$.

\smallskip\noindent
$(\Leftarrow)$
Assume that
\(
  \enc[\urgent]{\Delta}, \amod[!]{X} \vdash \amod[U]{a}
\).
Since $\enc[\urgent]{\Delta, X} \vdash \amod[U]{e}$ for all $e \in X$.
Take any $\sigma$ such that $\mkset{\sigma} = X$.
By Lemma~\ref{lem:pcl:proof-urgent} it follows that
there exist $\sigma,\eta$ such that $\mkset{\sigma} = X$ and
$\sigma a \eta \in \sem{\Delta, X}$.
By Def.~\ref{def:pcl:urgent}, we conclude that
$\atom{a} \in \curgent[\Delta]{X}$.
\end{proof}

\section{A logical view of contracts} \label{ces-pcl}

In this section we present our main results about the relation
between contracts and \pcl.
For a conflict-free \CES\ $\esname$ and a Horn \pcl\ theory $\Delta$,
we write $\Delta \sim \esname$ whenever there exists
an isomorphism 
which maps an enabling $X \vdash e$ in $\esname$
to a clause $(\bigwedge X) \imp e$ in $\Delta$,
and a circular enabling $X \Vdash e$
to $(\bigwedge X) \coimp e$.
Theorem~\ref{th:ces-pcl:agreement} shows that, 
for a relevant class of payoff functions, 
we can characterise agreement in terms of provability in \pcl.
Theorem~\ref{th:ces-pcl:prudence-prooftrace} states that 
proof traces correspond to sequences of prudent events.
Finally, Theorem~\ref{th:ces-pcl:winning-strategy} relates
winning strategies with urgent atoms.

Before providing the technical details,
we illustrate the relevance of these results
with the help of a couple of examples.

\begin{example} 
\label{ex:double-diamond}
Consider the following Horn \pcl\ theory $\Delta_{\star}$:
\begin{align*}
  \Delta_{\star} =
  \{
  &
  (e_0 \land e_1) \coimp e_6,\;
  e_6 \imp e_3,\;
  e_6 \imp e_4,\;
  e_3 \imp e_0, \\
  & 
  (e_4 \land e_5) \coimp e_7,\;
  e_7 \imp e_1,\;
  e_7 \imp e_2,\;
  e_2 \imp e_5
  \}
\end{align*}
It is possible to prove that $\Delta_{\star} \vdash e_i$ for all $i \in 0..7$.
However, this is not straightforward to see, and indeed
were any one of the $\coimp$ in $\Delta_{\star}$ replaced
with a $\imp$, then no atoms would have been provable.

We can exploit the correspondence between provability in \pcl\
and agreement in contracts to obtain a simple proof 
of $\Delta_{\star} \vdash e_i$.
To do that, observe that $\Delta_{\star}$
is isomorphic to the \CES\ $\esname_{\star}$ depicted as:
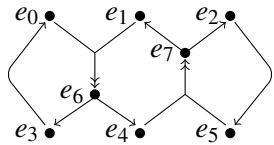
\begin{center}
\vcenteredhbox{%
\begin{tikzpicture}[scale=1.2]
\draw [fill=black]  (0,1.6)  circle (0.05)  node [left] {$e_0$};
\draw [fill=black]  (1, 1.6)  circle (0.05)  node [left] {$e_1$};
\draw [fill=black]  (2, 1.6)  circle (0.05)  node [left] {$e_2$};
\draw [fill=black]  (0,0.3)  circle (0.05)  node [left] {$e_3$};
\draw [fill=black]  (1,0.3)  circle (0.05)  node [left] {$e_4$};
\draw [fill=black]  (2,0.3)  circle (0.05)  node [left] {$e_5$};
\draw [fill=black]  (0.5, 0.73)  circle (0.05)  node [left] {$e_6$};
\draw [fill=black]  (1.5, 1.19 ) circle (0.05)  node  [left] {$e_7$};
%
\draw [->>] (0.05,1.53)  --(0.5, 1.19) -- (0.5, 0.8);
\draw [-] (0.95,1.53) --(0.5,1.19); 
\draw [->] (0.5,0.7) --(0.05,0.39);  
\draw [->] (0.5,0.7) --(0.95,0.39);  
%
\draw [->>] (1.05,0.39) -- (1.5,0.73)-- (1.5, 1.12 ); 
\draw [-] (1.95,0.39) -- (1.5,0.73); 
\draw [->] (1.5, 1.19) -- (1.05,1.53);
\draw [->] (1.5, 1.19) -- (1.95,1.53);
%
\draw[->,  rounded corners] (-0.05,0.39) -- (-0.5, 0.95) -- (-0.05,1.53);
\draw[->,  rounded corners] (2.05,1.53) -- (2.5, 0.95) -- (2.05,0.39);
\end{tikzpicture}}
\end{center}
and let $\cname = \seq{\esname_{\star},\payoffsym}$,
where we assume a single participant {\pmv A},
whose payoff is 
\(
  \payoff{\pmv A}{} 
  = 
  \setcomp{\sigma}{\forall i \in 0..7.\; e_i \in \mkset{\sigma}}
\).

It is easy to check that the contract $\cname$ admits an agreement.
Indeed, $e_6$ and $e_7$ are prudent in $\epsilon$;
$e_0$ becomes prudent after $e_3$ is fired; $e_5$ after $e_2$;
events $e_3,e_4$ after $e_6$; events $e_1,e_2$ after $e_7$.
Therefore, there exists a winning strategy for {\pmv A} in $\cname$.
Theorem~\ref{th:ces-pcl:agreement} allows for transferring 
this result back to \pcl, by establishing that 
all the atoms $e_0,\ldots,e_7$ are provable in $\Delta_{\star}$

Furthermore, the correspondence between contracts and \pcl\
allows for easily constructing the proof traces of $\Delta_{\star}$ ---
which is not as straightforward by applying Def.~\ref{def:pcl:proof-trace}.
The plays $\sigma$ where {\pmv A} wins are those where 
only the prudent events are performed, i.e.:
\[
  \sigma 
  \;\in\;
  \big( \, e_6 \; (e_4 \mid e_3 e_0) \big) 
  \mid
  \big( \, e_7 \; (e_1 \mid e_2 e_5) \big) 
\]
By Theorem~\ref{th:ces-pcl:prudence-prooftrace},
these plays exactly correspond to
the proof traces of the \pcl\ theory $\Delta_{\star}$.
\end{example}

\begin{example}[Shy dancers]
\label{ex:shy-dancers}
There are $n^2$ guests at a wedding party, arranged in a grid of size $n \times n$.
The music starts, the guests would like to dance, 
but they are too timid to start.
Each guest will dance provided that at least other two guests in its 
8-cells neighborhood will do the same.

We model this scenario as follows.
For all $i,j \in 1..n$, 
${\pmv A}_{i,j}$ is the guest at cell $(i,j)$,
and $e_{i,j}$ is the event which models ${\pmv A}_{i,j}$ dancing.
The neighborhood of $(i,j)$ is
$I_{i,j} = \setcomp{(p,q) \neq (i,j)}{|p-i| \leq 1 \;\land\; |q-j| \leq 1}$,
and we define $E_{i,j} = \setcomp{e_{p,q}}{(p,q) \in I_{i,j}}$.
Let $\mathfrak{F}$ be the set of functions from
$\setenum{1..n} \times \setenum{1..n}$ to $\setenum{\vdash,\Vdash}$.
For all $\bullet \in \mathfrak{F}$, let $\esname^{\bullet}$ be the \CES:
\[
  \esname^{\bullet} \; = \; \textstyle \bigcup_{i,j \in 1..n} \esname_{i,j}^{\bullet}
  \hspace{32pt}
  \text{where } \;
  \esname_{i,j}^{\bullet} \; = \; 
  \setcomp{X \bullet(i,j) \; e_{i,j}}{X \subseteq E_{i,j} \;\land\; |X| = 2}
\]
Intuitively, each function $\bullet \in \mathfrak{F}$ establishes which guests use $\vdash$ and which use $\Vdash$.
For all $\bullet \in \mathfrak{F}$ and for all $i,j \in 1..n$,
guest ${\pmv A}_{i,j}$ promises to dance if at least two neighbours have already started (in case $\bullet(i,j) = \;\vdash$),
or under the guarantee that they will eventually dance 
(when $\bullet(i,j) = \;\Vdash$).

Now, let $\payoff{({\pmv A}_{i,j})} = \setcomp{\sigma}{\mkset{\sigma} \cap E_{i,j} \geq 2}$, for all $i,j \in 1..n$.
For all $\bullet \in \mathfrak{F}$, we ask whether the contract 
$\cname^{\bullet} = \seq{\esname^{\bullet},\payoffsym}$
admits an agreement,
i.e.\ if all guests will eventually dance.
We have that $\cname^{\bullet}$ admits an agreement
iff there exist two guests in the same neighborhood which use $\Vdash$.
Formally:
\begin{equation*} 
  \exists i,j \in 1..n.\;\;
  \exists (p,q), (p',q') \in I_{i,j}.\;\;
  (p,q) \neq (p',q')
  \;\land\;
  \bullet(p,q) = \;\Vdash\; = \bullet(p',q')
\end{equation*}
Indeed, when
the above holds,
the strategy:
\[
  \Sigma^{\bullet}_{i,j}(\sigma) \; = \; 
  \begin{cases}
    \setenum{e_{i,j}} & \text{if $e_{i,j} \not\in \mkset{\sigma}$, and $\bullet(i,j) = \,\Vdash$ or $\mkset{\sigma} \vdash e_{i,j}$} \\
    \emptyset & \text{otherwise}
  \end{cases}
\]
is winning, for all guests ${\pmv A}_{i,j}$.
As noted in the previous example, the correspondence established
by Theorem~\ref{th:ces-pcl:agreement} allows us to transfer the above
observations to \pcl.
In particular, the above provides a simple proof that, 
in the Horn \pcl\ theory:
\[
  \Delta^{\bullet} \; = \; 
  \setcomp
  {\alpha \bullet(i,j) \; e_{i,j}}
  {\mkset{\alpha} \subseteq E_{i,j} \;\land\; |\mkset{\alpha}| \geq 2 \;\land\; i,j \in 1..n}
\]
some atom is provable iff there exist at least two distinct 
clauses which use $\coimp$.
Again, this result would not be easy to prove directly, without exploiting
the correspondence between agreements and provability.
\end{example}

\begin{definition}
We write $\Delta \sim \esname$ when $\esname$ is conflict-free, and
\[
\textstyle
  \Delta 
  \; = \;
  \setcomp{(\bigwedge X) \imp e}{X \vdash e \in \esname} 
  \cup
  \setcomp{(\bigwedge X) \coimp e}{X \Vdash e \in \esname}
\]
\end{definition}

To relate agreement with provability, 
we consider the class of \emph{reachability payoffs},
which neglect the order in which events are performed.
This class is quite broad.
For instance, it includes the \emph{offer-request payoffs}~\cite{BCZ13post}.
Intuitively, these are used by participants which 
want to be paid for each provided service.
Each participant {\pmv A} has 
a set $\setenum{O_{\pmv A}^0 , O_{\pmv A}^1, \ldots}$ of sets of events
(the \emph{offers}),
and a corresponding set $\setenum{R_{\pmv A}^0, R_{\pmv A}^1,\ldots}$ 
(the \emph{requests}).
To be winning, whenever {\pmv A} performs in a play
some offer $O_{\pmv A}^i$ (in whatever order), 
the play must also contain 
the corresponding request $R_{\pmv A}^i$,
and at least one of the requests has to be fulfilled.

\begin{definition}
\label{def:ces-pcl:reachability-payoff}
A \emph{reachability payoff} is a function $\payoffsym: \aname \rightarrow \powset{E^{\infty}}$ 
such that if $\mkset{\sigma} = \mkset{\eta}$ then 
$\sigma \in \payoff{\pmv A}{} \iff \eta \in \payoff{\pmv A}{}$,
for all ${\pmv A} \in \aname$.
\end{definition}

Alternatively, $\payoffsym$ is a reachability payoff
when there exists some predicate $\varphi \subseteq \powset{E}$
such that $\sigma \in \payoff{\pmv A}{}$ iff $\mkset{\sigma} \in \varphi$,
for all ${\pmv A} \in \aname$.

\medskip
The following theorem gives a logical characterisation of agreements.
If $\payoffsym$ is a reachability payoff induced by $\varphi$,
and $\Delta \sim \esname$,
then the contract $\seq{\esname,\payoffsym}$
admits an agreement whenever the set provable atoms in $\Delta$ 
satisfies the predicate $\varphi$.

\begin{restatable}{mythm}{agreement}
\label{th:ces-pcl:agreement}
Let $\Delta \sim \esname$,
and let $\payoffsym$ be a reachability payoff defined by the predicate $\varphi$.
Then, the contract $\cname = \seq{\esname,\payoffsym}$ admits an agreement iff 
$\setcomp{a}{\Delta \vdash a} \in \varphi$.
\end{restatable}

\begin{example} \label{ex:ces-pcl:or-payoffs}
Consider the following offer-request payoff $\payoffsym$ of {\pmv A} and {\pmv B}:
\[
\begin{array}{llll}
  O_{\pmv A}^0 = \setenum{a_0}
  \qquad
  &
  O_{\pmv A}^1 = \setenum{a_0,a_1} 
  \qquad
  &
  O_{\pmv B}^0 = \setenum{b_0} 
  \qquad
  &
  O_{\pmv B}^1 = \setenum{b_2} 
  \\
  R_{\pmv A}^0 = \setenum{b_0,b_2}
  &
  R_{\pmv A}^1 = \setenum{b_1}
  &
  R_{\pmv B}^0 = \setenum{a_0}
  \quad
  &
  R_{\pmv B}^1 = \setenum{a_0,a_2}
\end{array}
\]
and let the obligations of {\pmv A} and {\pmv B} be modelled by the \CES\ $\esname$ with enablings:
\[
  \setenum{b_0,b_2} \Vdash a_0
  \quad
  b_1 \vdash a_1
  \quad
  b_2 \Vdash a_2
  \quad
  a_0 \vdash b_0
  \quad
  \setenum{a_0,a_1} \vdash b_1
  \quad
  \setenum{a_0,a_2} \vdash b_2
\]
In the \pcl\ theory $\Delta \sim \esname$, the set of provable atoms is
$\setenum{a_0,a_2,b_0,b_2}$.
Therefore, by Theorem~\ref{th:ces-pcl:agreement} it follows that
the contract $\cname = \seq{\esname,\payoffsym}$ admits an agreement.
\end{example}

Recall from Def.~\ref{def:ces:agreement} that, 
when a contract admits an agreement,
all participants have a winning strategy.
Two relevant question are then how to construct a winning strategy 
for each participant, and how such strategy is related to \pcl.
We answer these questions in Theorem~\ref{th:ces-pcl:winning-strategy} below,
where we show that
a winning strategy can be obtained by following the order of urgent atoms.

In order to prove Theorem~\ref{th:ces-pcl:winning-strategy} we need 
to establish some further results about strategies and proof traces.
The first result is Lemma~\ref{lem:ces-pcl:prudentcf},
which provides an alternative characterisation of prudent events 
in case of conflict-free contracts.
We denote with $\preach{X}$ the set \emph{reachable events with past $X$}.
Intuitively, if a set $X$ of events has been performed in the past,
we consider an event $e \not\in X$ reachable with past $X$ 
when $e$ occurs in some play $\sigma\eta$
where the prefix $\sigma$ is a linearization of $X$, and 
the overall credits are contained in $X$
(\text{i.e.}, past debits need not be honoured).
Lemma~\ref{lem:ces-pcl:prudentcf} states that an event 
$e$ is prudent for {\pmv A} in~$\sigma$ whenever 
$e \in \cprudent{\mkset{\sigma}}$, 
namely the set of events
which are $\vdash$-enabled by $\mkset{\sigma}$, or 
$\Vdash$-enabled by $\mkset{\sigma} \cup \preach{\mkset{\sigma}}$.

\begin{restatable}{mylem}{prudentcf}
\label{lem:ces-pcl:prudentcf}
For a set $X \subseteq E$, let
\begin{align*}
  \preach{X}
  & = 
  \setcomp{e \not\in X}{\exists \sigma, \eta : \;
    \mkset{\sigma} = X, \; 
    e \in \mkset{\eta},\, \text{and }
    \mkcredit{\sigma\eta } \subseteq X}
  \\
  \cprudent{X} 
  & =
  \setcomp
  {e \not\in X\!}
  {\! X \vdash e \text{ or }
   X \cup \preach{X} \Vdash e}
\end{align*}
Then, $e$ is prudent in $\sigma$ iff $e \in \cprudent{\mkset{\sigma}}$.
\end{restatable}

The criterion given by Lemma~\ref{lem:ces-pcl:prudentcf}
is much simpler than
the mutually coinductive definition of prudence in Def.~\ref{def:ces:prudence}.
Indeed, a polynomial-time algorithm for computing $\cprudent{X}$
can be easily devised as follows.
At step 0, we compute the reflexive transitive closure $X_1$ 
of the hypergraph of the \CES\ $\esname$ (neglecting the $\Vdash$-enablings),
taking as start nodes all the events in $X \cup Y_0$, where
$Y_0 = \setcomp{e}{\exists Z : Z \Vdash e \in \esname}$
contains the events at the right of some $\Vdash$ in $\esname$.
The transitive closure can be computed 
in polynomial time in the number of events in~$\esname$.
If $X_1 \Vdash Y_0$, then $X_1 = \preach{X}$.
Otherwise, we repeat the above procedure with start nodes $X \cup Y_1$,
where $Y_1 = \setcomp{e \in Y_0}{X_1 \Vdash e}$,
until reaching a fixed point.
In the worst case, we do $n$ steps, hence we have a 
polynomial algorithm for computing $\preach{X}$.
After this, $\cprudent{X}$ can be easily computed,
as in Lemma~\ref{lem:ces-pcl:prudentcf}.

The following lemma provides a link between contracts and \pcl,
by establishing that prudent events in a \CES\ $\esname$
correspond exactly to urgent atoms in a \pcl\ theory $\Delta \sim \esname$.
The idea of the proof is to exploit
the mapping $\enc[\urgent]{\,}$ in Def.~\ref{def:pcl:enc-u}
as a bridge between \CES\ and \pcl.
To do that, we first map $\esname$ to a \pcl\ theory $\enc[\urgent]{\esname}$,
and we relate the prudent events in $\esname$ to the provable atoms 
in $\enc[\urgent]{\esname}$.
Since $\Delta \sim \esname$, we have that
$\enc[\urgent]{\esname} = \enc[\urgent]{\Delta}$,
and so by Theorem~\ref{th:pcl:enc-u} we can relate
provability in $\enc[\urgent]{\Delta}$ with
urgent atoms in $\Delta$.
Summing up, the prudent events in $\esname$ are the urgent atoms in $\Delta$.

\begin{restatable}{mylem}{prudenturgent}
\label{lem:ces-pcl:prudent-urgent}
Let $\Delta \sim \esname$. 
Then, for all $X \subseteq E$, 
\(
  \cprudent[\esname]{X} = \curgent[\Delta]{X}
\).
\end{restatable}

We can now relate prudence in contracts with proofs in \pcl.
Theorem~\ref{th:ces-pcl:prudence-prooftrace} states that
the plays of prudent events correspond to prefixes of proof traces.

\begin{restatable}{mythm}{prudenceprooftrace}
 \label{th:ces-pcl:prudence-prooftrace}
Say $\sigma = \seq{e_0 \, e_1 \cdots}$ is a \emph{prudent play} of $\esname$
when $e_i$ is prudent for~$\sigma_i$ in $\esname$, for all $i$.
If $\Delta \sim \esname$, then 
$\sigma$ is a prudent play of $\esname$
iff
\(
  \exists \eta.\; \sigma \eta \in \sem{\Delta}
\).
\end{restatable}

\begin{example}
The prudent plays of the \CES\ $\esname_3$ in Fig.~\ref{fig:ces}
are $\epsilon$, $a$, and $ab$ (see Ex.~\ref{ex:ces:prudence}).
By Theorem~\ref{th:ces-pcl:prudence-prooftrace}, these can be extended
to proof traces in the corresponding Horn \pcl\ theory $\Delta_3 \sim \esname_3$.
Indeed, $ab$ is a proof trace of $\Delta_3$ (see Ex.~\ref{ex:pcl:proof-trace}).
\end{example}

Our last main result relates the winning strategies of a contract 
$\cname = \seq{\esname,\payoffsym}$ with the proof traces of
a \pcl\ theory $\Delta \sim \esname$.
In particular, for all participants~{\pmv A} we construct a strategy that,
in a play $\sigma$, enables exactly the events of {\pmv A} 
which are urgent in $\mkset{\sigma}$.
This strategy is prudent for {\pmv A}, and leads {\pmv A} to a 
winning play whenever {\pmv A} agrees on $\cname$.

\begin{restatable}{mythm}{winningstrategy}
\label{th:ces-pcl:winning-strategy}
Let $\Delta \sim \esname$, and
let the strategy $\Sigma_{\pmv A}$ be defined as:
\[
  \Sigma_{\pmv A}(\sigma) 
  \;\; = \;\;
  \curgent[\Delta]{\mkset{\sigma}} \;\cap\; \invprinc{\pmv A}
\]
Then, $\Sigma_{\pmv A}$ is a prudent strategy for {\pmv A} in  
$\cname = \seq{\esname,\payoffsym}$.
Moreover, if $\payoffsym$ is a reachability payoff
and $\cname$ admits an agreement,
then $\Sigma_{\pmv A}$ is winning for {\pmv A}.
\end{restatable}

\section{Conclusions} \label{sect:conclusions}

We have studied the relations between two foundational models for contracts.
The main result is that the notions of agreement and winning strategy
in the game-theoretic model of~\cite{BCZ13post} have been related,
respectively, to that of provability and proof traces 
in the logical model of~\cite{BZ10lics} 
(Theorems~\ref{th:ces-pcl:agreement} and~\ref{th:ces-pcl:winning-strategy}).

Some preliminary work on relating event structures 
with the logic \pcl\ has been reported in~\cite{BCPZ12places}.
The model of~\cite{BCPZ12places} does not exploit game-theoretic notions:
payoffs are just sets of events,
and agreement is defined as the existence of a configuration
in the \CES\ which contains such set.
In this simplified model, it is shown that an event is reachable in
a \CES\ whenever it is provable in the corresponding \pcl\ theory.
Hence, an agreement exists whenever all the events in the participant payoffs are provable.
Theorem~\ref{th:ces-pcl:agreement} extends this result to a more general 
(game-theoretic) notion of agreement and of payoff.

In~\cite{BCP13fsen} the idea of performing events ``on credit'' has been
explored in the domain of Petri nets.
In the variant of Petri nets presented in~\cite{BCP13fsen} 
(called Lending Petri nets, LPNs), 
certain places may be tagged as ``lending'', with the meaning that
their marking can become negative, but must be eventually brought back to 
a nonnegative value.
A technique is presented to transform Horn \pcl\ theories into LPNs,
and it is shown that provability in a \pcl\ theory 
corresponds to \emph{weak termination} in the LPN obtained by the transformation.

An encoding of Horn \pcl\ formulae into a variant of CCS has been presented in~\cite{BTZ12sacs}.
Very roughly, the encoding maps a clause $\alpha \imp a$ in a process
which inputs all the channels in $\alpha$ and then outputs on $a$, 
while a clause $\alpha \coimp a$
the input of $\alpha$ and the output of $a$ is done in parallel.
The actual encoding is more sophisticated than the above intuition,
mostly because it has to take into account multiple participants
which share the same channels, and it has to 
preserve the notion of culpability defined in the logical model.
In particular, in the CCS model a participant has to be culpable 
only when omitting to produce a promised output,
or omitting to input an available message.

\paragraph{Acknowledgments.}
Work partially supported by 
Aut. Region of Sardinia grants 
L.R.7/2007 CRP-17285 (TRICS),
P.I.A.\ 2010 Project ``Social Glue'',
by MIUR PRIN 2010-11 project ``Security Horizons'',
and by COST Action IC1201: Behavioural Types for Reliable Large-Scale Software Systems (BETTY).

\bibliographystyle{eptcs}
\bibliography{main}

\begin{thebibliography}{10}
\providecommand{\bibitemdeclare}[2]{}
\providecommand{\surnamestart}{}
\providecommand{\surnameend}{}
\providecommand{\urlprefix}{Available at }
\providecommand{\url}[1]{\texttt{#1}}
\providecommand{\href}[2]{\texttt{#2}}
\providecommand{\urlalt}[2]{\href{#1}{#2}}
\providecommand{\doi}[1]{doi:\urlalt{http://dx.doi.org/#1}{#1}}
\providecommand{\bibinfo}[2]{#2}

\bibitemdeclare{article}{Aalst10cj}
\bibitem{Aalst10cj}
\bibinfo{author}{Wil M.~P. \surnamestart van~der Aalst\surnameend},
  \bibinfo{author}{Niels \surnamestart Lohmann\surnameend},
  \bibinfo{author}{Peter \surnamestart Massuthe\surnameend},
  \bibinfo{author}{Christian \surnamestart Stahl\surnameend} \&
  \bibinfo{author}{Karsten \surnamestart Wolf\surnameend}
  (\bibinfo{year}{2010}): \emph{\bibinfo{title}{Multiparty Contracts: Agreeing
  and Implementing Interorganizational Processes}}.
\newblock {\sl \bibinfo{journal}{Comput. J.}}
  \bibinfo{volume}{53}(\bibinfo{number}{1}), pp. \bibinfo{pages}{90--106},
  \doi{10.1093/comjnl/bxn064}.

\bibitemdeclare{article}{Abadi93access}
\bibitem{Abadi93access}
\bibinfo{author}{Mart{\'\i}n \surnamestart Abadi\surnameend},
  \bibinfo{author}{Michael \surnamestart Burrows\surnameend},
  \bibinfo{author}{Butler \surnamestart Lampson\surnameend} \&
  \bibinfo{author}{Gordon \surnamestart Plotkin\surnameend}
  (\bibinfo{year}{1993}): \emph{\bibinfo{title}{A calculus for access control
  in distributed systems}}.
\newblock {\sl \bibinfo{journal}{{ACM} {TOPLAS}}}
  \bibinfo{volume}{4}(\bibinfo{number}{15}), pp. \bibinfo{pages}{706--734},
  \doi{10.1145/155183.155225}.

\bibitemdeclare{article}{Abadi93logical}
\bibitem{Abadi93logical}
\bibinfo{author}{Mart\'{\i}n \surnamestart Abadi\surnameend} \&
  \bibinfo{author}{Gordon~D. \surnamestart Plotkin\surnameend}
  (\bibinfo{year}{1993}): \emph{\bibinfo{title}{A Logical View of
  Composition}}.
\newblock {\sl \bibinfo{journal}{Theoretical Computer Science}}
  \bibinfo{volume}{114}(\bibinfo{number}{1}), pp. \bibinfo{pages}{3--30},
  \doi{10.1016/0304-3975(93)90151-I}.

\bibitemdeclare{article}{Armbrust10cacm}
\bibitem{Armbrust10cacm}
\bibinfo{author}{Michael \surnamestart Armbrust\surnameend} et~al.
  (\bibinfo{year}{2010}): \emph{\bibinfo{title}{A view of cloud computing}}.
\newblock {\sl \bibinfo{journal}{Comm. ACM}}
  \bibinfo{volume}{53}(\bibinfo{number}{4}), pp. \bibinfo{pages}{50--58},
  \doi{10.1145/1721654.1721672}.

\bibitemdeclare{misc}{ces-pcl-long}
\bibitem{ces-pcl-long}
\bibinfo{author}{Massimo \surnamestart Bartoletti\surnameend},
  \bibinfo{author}{Tiziana \surnamestart Cimoli\surnameend},
  \bibinfo{author}{Paolo~Di \surnamestart Giamberardino\surnameend} \&
  \bibinfo{author}{Roberto \surnamestart Zunino\surnameend}:
  \emph{\bibinfo{title}{Contract agreements via logic}}.
\newblock \bibinfo{note}{Available online at
  \texttt{tcs.unica.it/papers/ces-pcl-long.pdf}}.

\bibitemdeclare{inproceedings}{BCP13fsen}
\bibitem{BCP13fsen}
\bibinfo{author}{Massimo \surnamestart Bartoletti\surnameend},
  \bibinfo{author}{Tiziana \surnamestart Cimoli\surnameend} \&
  \bibinfo{author}{G.~Michele \surnamestart Pinna\surnameend}
  (\bibinfo{year}{2013}): \emph{\bibinfo{title}{Lending {P}etri nets and
  contracts}}.
\newblock In: {\sl \bibinfo{booktitle}{Proc. {FSEN}}}, {\sl
  \bibinfo{series}{LNCS}} \bibinfo{volume}{8161},
  \bibinfo{publisher}{Springer}, \doi{10.1007/978-3-642-40213-5\_5}.

\bibitemdeclare{unpublished}{BCPZ13fi}
\bibitem{BCPZ13fi}
\bibinfo{author}{Massimo \surnamestart Bartoletti\surnameend},
  \bibinfo{author}{Tiziana \surnamestart Cimoli\surnameend},
  \bibinfo{author}{G.~Michele \surnamestart Pinna\surnameend} \&
  \bibinfo{author}{Roberto \surnamestart Zunino\surnameend}:
  \emph{\bibinfo{title}{Circular causality in event structures}}.
\newblock \bibinfo{note}{Submitted. Available online at
  \texttt{tcs.unica.it/papers/ces-long.pdf}. A preliminary version of this
  paper has been presented at {ICTCS} 2012}.

\bibitemdeclare{inproceedings}{BCPZ12places}
\bibitem{BCPZ12places}
\bibinfo{author}{Massimo \surnamestart Bartoletti\surnameend},
  \bibinfo{author}{Tiziana \surnamestart Cimoli\surnameend},
  \bibinfo{author}{G.~Michele \surnamestart Pinna\surnameend} \&
  \bibinfo{author}{Roberto \surnamestart Zunino\surnameend}
  (\bibinfo{year}{2012}): \emph{\bibinfo{title}{An event-based model for
  contracts}}.
\newblock In: {\sl \bibinfo{booktitle}{Proc. {PLACES}}},
  \doi{10.4204/EPTCS.109.3}.

\bibitemdeclare{inproceedings}{BCZ13post}
\bibitem{BCZ13post}
\bibinfo{author}{Massimo \surnamestart Bartoletti\surnameend},
  \bibinfo{author}{Tiziana \surnamestart Cimoli\surnameend} \&
  \bibinfo{author}{Roberto \surnamestart Zunino\surnameend}
  (\bibinfo{year}{2013}): \emph{\bibinfo{title}{A theory of agreements and
  protection}}.
\newblock In: {\sl \bibinfo{booktitle}{Proc. {POST}}}, {\sl
  \bibinfo{series}{LNCS}} \bibinfo{volume}{7796},
  \bibinfo{publisher}{Springer}, \doi{10.1007/978-3-642-36830-1\_10}.

\bibitemdeclare{article}{BTZ12sacs}
\bibitem{BTZ12sacs}
\bibinfo{author}{Massimo \surnamestart Bartoletti\surnameend},
  \bibinfo{author}{Emilio \surnamestart Tuosto\surnameend} \&
  \bibinfo{author}{Roberto \surnamestart Zunino\surnameend}
  (\bibinfo{year}{2012}): \emph{\bibinfo{title}{Contract-oriented Computing in
  {CO}${}_2$}}.
\newblock {\sl \bibinfo{journal}{Scientific Annals in Computer Science}}
  \bibinfo{volume}{22}(\bibinfo{number}{1}), pp. \bibinfo{pages}{5--60},
  \doi{10.7561/SACS.2012.1.5}.

\bibitemdeclare{inproceedings}{BZ10lics}
\bibitem{BZ10lics}
\bibinfo{author}{Massimo \surnamestart Bartoletti\surnameend} \&
  \bibinfo{author}{Roberto \surnamestart Zunino\surnameend}
  (\bibinfo{year}{2010}): \emph{\bibinfo{title}{A Calculus of Contracting
  Processes}}.
\newblock In: {\sl \bibinfo{booktitle}{{LICS}}}, \doi{10.1109/LICS.2010.25}.

\bibitemdeclare{inproceedings}{bhty10}
\bibitem{bhty10}
\bibinfo{author}{Laura \surnamestart Bocchi\surnameend}, \bibinfo{author}{Kohei
  \surnamestart Honda\surnameend}, \bibinfo{author}{Emilio \surnamestart
  Tuosto\surnameend} \& \bibinfo{author}{Nobuko \surnamestart
  Yoshida\surnameend} (\bibinfo{year}{2010}): \emph{\bibinfo{title}{A theory of
  design-by-contract for distributed multiparty interactions}}.
\newblock In: {\sl \bibinfo{booktitle}{CONCUR}},
  \doi{10.1007/978-3-642-15375-4\_12}.

\bibitemdeclare{inproceedings}{Bravetti08tgc}
\bibitem{Bravetti08tgc}
\bibinfo{author}{Mario \surnamestart Bravetti\surnameend},
  \bibinfo{author}{Ivan \surnamestart Lanese\surnameend} \&
  \bibinfo{author}{Gianluigi \surnamestart Zavattaro\surnameend}
  (\bibinfo{year}{2008}): \emph{\bibinfo{title}{Contract-Driven Implementation
  of Choreographies}}.
\newblock In: {\sl \bibinfo{booktitle}{Proc. {TGC}}}, pp.
  \bibinfo{pages}{1--18}, \doi{10.1007/978-3-642-00945-7\_1}.

\bibitemdeclare{inproceedings}{Bravetti07fsen}
\bibitem{Bravetti07fsen}
\bibinfo{author}{Mario \surnamestart Bravetti\surnameend} \&
  \bibinfo{author}{Gianluigi \surnamestart Zavattaro\surnameend}
  (\bibinfo{year}{2007}): \emph{\bibinfo{title}{Contract Based Multi-party
  Service Composition}}.
\newblock In: {\sl \bibinfo{booktitle}{Proc. {FSEN}}}, pp.
  \bibinfo{pages}{207--222}, \doi{10.1007/978-3-540-75698-9\_14}.

\bibitemdeclare{article}{Castagna09toplas}
\bibitem{Castagna09toplas}
\bibinfo{author}{Giuseppe \surnamestart Castagna\surnameend},
  \bibinfo{author}{Nils \surnamestart Gesbert\surnameend} \&
  \bibinfo{author}{Luca \surnamestart Padovani\surnameend}
  (\bibinfo{year}{2009}): \emph{\bibinfo{title}{A theory of contracts for Web
  services}}.
\newblock {\sl \bibinfo{journal}{{ACM} {TOPLAS}}}
  \bibinfo{volume}{31}(\bibinfo{number}{5}), \doi{10.1145/1538917.1538920}.

\bibitemdeclare{inproceedings}{bpel4chor}
\bibitem{bpel4chor}
\bibinfo{author}{G.~\surnamestart Decker\surnameend},
  \bibinfo{author}{O.~\surnamestart Kopp\surnameend},
  \bibinfo{author}{F.~\surnamestart Leymann\surnameend} \&
  \bibinfo{author}{M~\surnamestart Weske\surnameend} (\bibinfo{year}{2007}):
  \emph{\bibinfo{title}{{BPEL4Chor}: Extending {BPEL} for Modeling
  Choreographies}}.
\newblock In: {\sl \bibinfo{booktitle}{Proc. {ICWS}}}.

\bibitemdeclare{article}{Gelati04normative}
\bibitem{Gelati04normative}
\bibinfo{author}{Jonathan \surnamestart Gelati\surnameend},
  \bibinfo{author}{Antonino \surnamestart Rotolo\surnameend},
  \bibinfo{author}{Giovanni \surnamestart Sartor\surnameend} \&
  \bibinfo{author}{Guido \surnamestart Governatori\surnameend}
  (\bibinfo{year}{2004}): \emph{\bibinfo{title}{Normative autonomy and
  normative co-ordination: Declarative power, representation, and mandate}}.
\newblock {\sl \bibinfo{journal}{Artificial Intelligence and Law}}
  \bibinfo{volume}{12}(\bibinfo{number}{1-2}), pp. \bibinfo{pages}{53--81},
  \doi{10.1007/s10506-004-1922-2}.

\bibitemdeclare{inproceedings}{Hildebrandt10places}
\bibitem{Hildebrandt10places}
\bibinfo{author}{Thomas~T. \surnamestart Hildebrandt\surnameend} \&
  \bibinfo{author}{Raghava~Rao \surnamestart Mukkamala\surnameend}
  (\bibinfo{year}{2010}): \emph{\bibinfo{title}{Declarative Event-Based
  Workflow as Distributed Dynamic Condition Response Graphs}}.
\newblock In: {\sl \bibinfo{booktitle}{Proc. {PLACES}}},
  \doi{10.4204/EPTCS.69.5}.

\bibitemdeclare{incollection}{scribble}
\bibitem{scribble}
\bibinfo{author}{Kohei \surnamestart Honda\surnameend}, \bibinfo{author}{Aybek
  \surnamestart Mukhamedov\surnameend}, \bibinfo{author}{Gary \surnamestart
  Brown\surnameend}, \bibinfo{author}{Tzu-Chun \surnamestart Chen\surnameend}
  \& \bibinfo{author}{Nobuko \surnamestart Yoshida\surnameend}
  (\bibinfo{year}{2011}): \emph{\bibinfo{title}{Scribbling Interactions with a
  Formal Foundation}}.
\newblock In: {\sl \bibinfo{booktitle}{Distributed Computing and Internet
  Technology}}, {\sl \bibinfo{series}{LNCS}} \bibinfo{volume}{6536},
  \bibinfo{publisher}{Springer}, pp. \bibinfo{pages}{55--75},
  \doi{10.1007/978-3-642-19056-8\_4}.

\bibitemdeclare{inproceedings}{Honda08popl}
\bibitem{Honda08popl}
\bibinfo{author}{Kohei \surnamestart Honda\surnameend}, \bibinfo{author}{Nobuko
  \surnamestart Yoshida\surnameend} \& \bibinfo{author}{Marco \surnamestart
  Carbone\surnameend} (\bibinfo{year}{2008}): \emph{\bibinfo{title}{Multiparty
  asynchronous session types}}.
\newblock In: {\sl \bibinfo{booktitle}{POPL}}, pp. \bibinfo{pages}{273--284},
  \doi{10.1145/1328438.1328472}.

\bibitemdeclare{misc}{wscdl}
\bibitem{wscdl}
\bibinfo{author}{N.~\surnamestart Kavantzas\surnameend},
  \bibinfo{author}{D.~\surnamestart Burdett\surnameend},
  \bibinfo{author}{G.~\surnamestart Ritzinger\surnameend},
  \bibinfo{author}{T.~\surnamestart Fletcher\surnameend},
  \bibinfo{author}{Y.~\surnamestart Lafon\surnameend} \&
  \bibinfo{author}{C~\surnamestart Barreto\surnameend} (\bibinfo{year}{2005}):
  \emph{\bibinfo{title}{{Web} {Services} {Choreography} {Description}
  {Language} V.~1.0}}.

\bibitemdeclare{article}{Lomuscio11fi}
\bibitem{Lomuscio11fi}
\bibinfo{author}{Alessio \surnamestart Lomuscio\surnameend},
  \bibinfo{author}{Wojciech \surnamestart Penczek\surnameend},
  \bibinfo{author}{Monika \surnamestart Solanki\surnameend} \&
  \bibinfo{author}{Maciej \surnamestart Szreter\surnameend}
  (\bibinfo{year}{2011}): \emph{\bibinfo{title}{Runtime Monitoring of Contract
  Regulated Web Services}}.
\newblock {\sl \bibinfo{journal}{Fundam. Inform.}}
  \bibinfo{volume}{111}(\bibinfo{number}{3}), \doi{10.3233/FI-2011-566}.

\bibitemdeclare{article}{Prisacariu11jlap}
\bibitem{Prisacariu11jlap}
\bibinfo{author}{Cristian \surnamestart Prisacariu\surnameend} \&
  \bibinfo{author}{Gerardo \surnamestart Schneider\surnameend}
  (\bibinfo{year}{2012}): \emph{\bibinfo{title}{A Dynamic Deontic Logic for
  Complex Contracts}}.
\newblock {\sl \bibinfo{journal}{The Journal of Logic and Algebraic Programming
  ({JLAP})}} \bibinfo{volume}{81}(\bibinfo{number}{4}),
  \doi{10.1016/j.jlap.2012.03.003}.

\bibitemdeclare{inproceedings}{Raimondi08fse}
\bibitem{Raimondi08fse}
\bibinfo{author}{Franco \surnamestart Raimondi\surnameend},
  \bibinfo{author}{James \surnamestart Skene\surnameend} \&
  \bibinfo{author}{Wolfgang \surnamestart Emmerich\surnameend}
  (\bibinfo{year}{2008}): \emph{\bibinfo{title}{Efficient online monitoring of
  web-service SLAs}}.
\newblock In: {\sl \bibinfo{booktitle}{{SIGSOFT} {FSE}}},
  \doi{10.1145/1453101.1453125}.

\bibitemdeclare{article}{Statman79pspace}
\bibitem{Statman79pspace}
\bibinfo{author}{Richard \surnamestart Statman\surnameend}
  (\bibinfo{year}{1979}): \emph{\bibinfo{title}{Intuitionistic propositional
  logic is polynomial-space complete}}.
\newblock {\sl \bibinfo{journal}{Theoretical Computer Science}}
  \bibinfo{volume}{9}, pp. \bibinfo{pages}{67--72},
  \doi{10.1016/0304-3975(79)90006-9}.

\bibitemdeclare{inproceedings}{Winskel86}
\bibitem{Winskel86}
\bibinfo{author}{Glynn \surnamestart Winskel\surnameend}
  (\bibinfo{year}{1986}): \emph{\bibinfo{title}{Event Structures}}.
\newblock In: {\sl \bibinfo{booktitle}{Advances in {Petri} Nets}}, pp.
  \bibinfo{pages}{325--392}, \doi{10.1007/3-540-17906-2\_31}.

\end{thebibliography}

\iftoggle{proofs}{%
  \appendix
  \newpage
  \renewenvironment{theorem}{\begin{appthm}}{\end{appthm}}
  \renewenvironment{lemma}{\begin{applem}}{\end{applem}}
  \renewenvironment{definition}{\begin{appdef}}{\end{appdef}}
  \input{proofs-ces.tex}
  \input{proofs-pcl.tex}
  \input{proofs-ces-pcl.tex}
}{%
}


\end{document}